\documentclass[12pt,letterpaper]{article}
\usepackage[utf8]{inputenc}
\usepackage[english]{babel}
\usepackage[margin=1in]{geometry} 
\usepackage{amsmath,amsthm,amssymb,amsfonts,bbm}
\usepackage{titling}
\usepackage{graphicx}
\usepackage{verbatim}
\usepackage{xcolor}
\usepackage{calc}
\usepackage[ruled]{algorithm2e}   %algorithm
\usepackage{bbm}
\usepackage{mathrsfs}
\usepackage{tikz}
\usepackage{authblk}

\usepackage{caption}  % subfigures
\usepackage{subcaption} % subfigures

\usepackage{hyperref}
\hypersetup{colorlinks=true, linkcolor=darkred, citecolor=darkgreen, urlcolor=darkblue} % package for hyperref
\usepackage[round]{natbib} % package for author year citation

\usepackage{booktabs} % for \cmidrule(lr){}
\usepackage{arydshln} % to include horizaontal dash rule in tables

\usepackage{mathtools}
\usepackage{rotating}
\usepackage{tikz}
\usetikzlibrary{positioning}

\usepackage{color}
\definecolor{darkred}{RGB}{100,0,0}
\definecolor{darkgreen}{RGB}{0,100,0}
\definecolor{darkblue}{RGB}{0,0,150}

\usepackage{pagecolor}
\definecolor{softgreen}{RGB}{230, 245, 230}
\usepackage{amsmath}
\usepackage{amsfonts}
\usepackage{amssymb}
\usepackage{bbm}
\newcommand{\ind}{\mathbbm{1}}
\usepackage{booktabs}

\usepackage{tikz} % for DAG
\usetikzlibrary{positioning, shapes.geometric} % <--- important!
\usetikzlibrary{positioning}
\usetikzlibrary{arrows.meta}
\usepackage{color}
\definecolor{darkred}{RGB}{100,0,0}
\definecolor{darkgreen}{RGB}{0,100,0}
\definecolor{darkblue}{RGB}{0,0,150}

\usepackage{pagecolor}
\definecolor{softgreen}{RGB}{230, 245, 230}
\usepackage{prodint} % for product integral \prodi

\usepackage{accents}

\newcommand{\E}{E}
\newcommand{\PP}{\mathbb{P}}

\newcommand{\M}{\mathbb{M}}

\newcommand{\A}{\mathcal{A}}
\newcommand{\B}{\mathcal{B}}

\newcommand{\Ec}{\mathcal{E}}

\newcommand{\dd}{d}   % derivative
\newcommand{\btr}{b^*}  % the truth for b
\newcommand{\thetatr}{\theta}  % the true theta
\newcommand{\tauQ}{\tau}  % the maximum support of Q^*
\newcommand{\etaD}{\eta_D} 
\newcommand{\tauC}{\tau_{\max}}  % the maximum follow up time

\newcommand{\httc}{\hat\theta}  
\newcommand{\sigmac}{\sigma}

\newcommand{\convp}{\stackrel{p}{\rightarrow}}  %\xrightarrow[]{\ p\ }
\newcommand{\convd}{\stackrel{d}{\rightarrow}}  %\xrightarrow[]{\ d\ }

\newcommand{\bigCI}{\mathrel{\text{\scalebox{1.07}{$\perp\mkern-10mu\perp$}}}}

\usepackage{mathtools}

\newtheorem{theorem}{Theorem}
\newtheorem{proposition}{Proposition}

\newtheorem{lemma}{Lemma}
\newtheorem{assumption}{Assumption}
\newtheorem{remark}{Remark}

\usepackage{xr}            % For \externaldocument
\makeatletter
\newcommand*{\addFileDependency}[1]{% argument=file name and extension
	\typeout{(#1)}
	\@addtofilelist{#1}
	\IfFileExists{#1}{}{\typeout{No file #1.}}
}
\makeatother

\usepackage{setspace}

\title{\LARGE Proximal Survival Analysis for Dependent Left Truncation}

\author[1]{\small Yuyao Wang\thanks{Current address: Department of Public Health and Health Sciences, Northeastern University, 360 Huntington Ave, Boston, MA 02115, USA. Email: yuya.wang@northeastern.edu}}
\author[2]{Andrew Ying\thanks{To whom correspondence should be addressed. Email: aying9339@gmail.com}}
\author[1,3]{Ronghui Xu}

\affil[1]{\small Department of Mathematics, University of California San Diego}
\affil[2]{\small Irvine, CA}
\affil[3]{\small Herbert Wertheim School of Public Health and Halicioglu Data Science Institute,\\ University of California San Diego}
\date{ }

\begin{document}

\begin{singlespace}
\vspace{-3em}
\maketitle
\vspace{-3em}
\begin{abstract}
In prevalent cohort studies with delayed entry, time-to-event outcomes are often subject to left truncation where only subjects that have not experienced the event at study entry are included, leading  to selection bias. Existing methods for handling left truncation mostly rely on the (quasi-)independence assumption or the weaker conditional (quasi-)independence assumption which assumes that conditional on observed covariates, the left truncation time and the event time are independent on the observed region. In practice, however, our analysis of the Honolulu Asia Aging Study (HAAS) suggests that the conditional quasi-independence assumption may fail because measured covariates often serve only as imperfect proxies for the underlying mechanisms, such as latent health status, that induce dependence between truncation and event times. To address this gap, we propose a proximal weighting identification framework that admits these factors may not be fully observed. We then construct an estimator based on the framework and study its asymptotic properties. We examine the finite sample performance of the proposed estimator by comprehensive simulations, and apply it to analyzing the cognitive impairment-free survival probabilities using data from the Honolulu Asia Aging Study.
\end{abstract}
\end{singlespace}

\newpage
\doublespacing

\section{Introduction}

Time-to-event outcomes are of interest in many clinical and epidemiological studies. Ideally, subjects would be enrolled at the time origin and followed long enough to observe the event of interest. However, this is often not feasible in practice. In many studies, there is unavoidable delayed entry, and only subjects who have not yet experienced the event are included.
For example, in aging studies, age is often the time scale of interest, for which birth is the time origin. However, following individuals from birth is extremely costly, especially when the events of interest, such as cognitive impairment or death, typically occur in late life. As another example, for pregnancy studies, women usually enroll in the study after the clinical recognition of their pregnancies, which often occurs weeks or even months after they got pregnant. Therefore, women with early pregnancy losses tend not to be included \citep{xu:cham, ying2024causal}. 
When only subjects whose event times greater than their left truncation times (i.e., study enrollment times in the above examples) are included, the time-to-event outcome is said to be subject to left truncation. In this setting, the observed data is a biased sample from the population of interest because subjects with later event times are preferentially selected. 
%\andrew{Add more on the consequences of ignoring selection bias, like biased estimate lead to wrong conclusion, decision making and waste of resources, etc.}
Without properly handling this selection bias, analysis based on the biased sample will result in biased estimates and invalid inferences, which may lead to misleading clinical or policy decisions and waste of resources.

To address the selection bias from left truncation, the random left truncation or the slightly weaker quasi-independent left truncation is conventionally assumed for estimating marginal quantities like the marginal survival probabilities \citep{woodroofe1985estimating, wang1986asymptotic, wang1989semiparametric, wang1991nonparametric, gross1996weighted, gross1996nonparametric, shen2010semiparametric, rabhi2021semiparametric}, where the left truncation time is assumed to be independent from the event time on the observed data region.
In regression settings, such as under Cox proportional hazards models, the requirement of the independence between the left truncation time and the event time can be relaxed to the conditional independence given covariates included in the regression model \citep{wang1993statistical, shen2009analyzing, qin2010statistical, qin2011maximum}. Such conditional independence assumption is usually referred to as covariate dependent left truncation. 
For estimating a marginal estimand under covariate dependent left truncation, inverse probability weighting \citep{vakulenko2022nonparametric} and doubly robust approaches \citep{wang2024doubly, wang2024learning} were recently developed.  
%In practice, the covariate dependent left truncation assumption may fail due to latent factors that influence both the left truncation and the event time of interest. These latent factors can be hard to measure directly, and measured covariates may only serve as imperfect proxies for them. 

% \andrew{Move a little bit earlier, after ``In practice, ... may fail..." Give example why it fails. (here is the problem gap) to handle it, borrow from ... }
In practice, however, the conditional quasi-independence assumption may fail. Our investigation into the Honolulu Asia Aging Study (HAAS, 1991-2012) \citep{p2012honolulu} suggests that the dependence between the study entry age and the event time may be induced by latent factors that are not fully captured by measured covariates. In the HAAS cohort, when analyzing cognitive impairment-free survival, the time-to-event is age to cognitive impairment or death. Since only subjects who were alive and had not developed cognitive impairment at HAAS entry were included, the time-to-event is left truncated by age at study entry. Latent factors such as socioeconomic status, health-seeking behavior, and overall health status are likely associated with both study entry age and  risks of mortality and cognitive impairment. Since these latent factors are difficult to measure directly, the standard (conditional) independence assumption required by existing approaches are likely violated, potentially leading to biased clinical conclusions.

This practical challenge motivates us to relax the the conditional independence assumption for left truncation by developing a proximal identification framework. Although latent factors like health status are unobserved, many measured covariates may serve as their imperfect proxies. For example, grip strength can be viewed as a proxy for overall physical frailty, while education and lifestyle variables (such as alcohol and cigarette consumptions) reflect socioeconomic status and health-seeking behavior.

Inspired by the proximal causal inference framework \citep{miao2018identifying, tchetgen2024introduction} and recent developments in proximal survival analysis for dependent censoring \citep{ying2024proximal}, we propose a proximal weighting identification framework for dependent left truncation. To our best knowledge, this is the first work to utilize proxy variables to handle unmeasured dependence in the context of left truncation. 
Following prior work in proximal inference, we 
assume that practitioners can correctly classify the measured covariates into three types: 
a) covariates that are directly associated with both the left truncation time and the event time; b) truncation proxies that are potentially (not necessarily) associated with the left truncation time and are associated with the event time only through the latent factors, for which the covariates are proxies; c) event time proxies that are potentially (not necessarily) associated with the event time and are associated with the left truncation time only through latent factors, for which the covariates are proxies.
In Section \ref{sec:preliminary}, we formally define the three types of proxies and provide a concrete example using HAAS data. 
Our framework leverages these proxies to obtain a truncation-inducing bridge process, defined as a solution to a set of integral equations.
We refer interested readers to \citet{tchetgen2024introduction} for an overview of the proximal framework.

Our framework leverages the above three types proxies to obtain a truncation-inducing bridge process which is defined as a solution to a set of integral equations indexed by time. Based on the bridge process, we propose a nonparametric identification and construct an % two-stage
%\lily{not sure a good term, we are not under learning; 2-stage often means ad hoc not accounting for 1st stage estimation error}\yuyao{What about "plug-in" or just "a estimator"?}
% \todo{After we finalize the wording here, update the wording for the rest of the paper}
estimator after the bridge process has been estimated. 
The  bridge process estimation is flexible and allows practitioners to incorporate the estimation approaches they prefer.
As an illustration, we consider a compatible semiparametric working model with additive form for the bridge process, under which a closed-form expression exists for the bridge process estimator.
We investigate generic assumptions on the bridge process estimator under which the proposed estimator for the parameter of interest is consistent and asymptotically normal. 
We further extend the proposed identification and estimation framework to allow right censoring.
The finite sample performance of the proposed estimator is studied by Monte Carlo simulations, and we apply the proposed method to analyze cognitive impairment-free survivals using data collected from HAAS. 
%\andrew{Your way of estiming first stage by additive model, like my paper, should be just one way but not the only way. I suggest you should set your tone as in my paper as the estimation and inference parts are purely one way of illustration, people don't have to go through additive models. It is more an example that this proxy method can work. But don't cut down other possibilities.}\yuyao{Updated above.}

The remainder of the paper is organized as follows. 
Section \ref{sec:preliminary} introduces the problem setup. 
Section \ref{sec:identification} introduces the key assumptions 
%needed 
and the proposed identification.
Section \ref{sec:estimation_asymptotics} discusses the estimation and the asymptotic properties of the proposed estimator under generic assumptions. 
Section \ref{sec:simu} includes simulation studies, and Section \ref{sec:application} includes an application using the HAAS data.
Section \ref{sec:discussion} concludes with a discussion.

\section{Problem setup}\label{sec:preliminary}
  
Let $Q^*$, $T^*$, $C^*$ denote the left truncation time, the event time of interest, and the right censoring time, respectively, in the full data population, that is, the population before left truncation. 
{A subject is observed only if $Q^*<T^*$. 
We use supscript asterisk `*' to denote variables in the full data population, that is, without left truncation; and without asterisk denote variables in the population under left truncation.}
We are interested in estimating the expectation of an arbitrarily transformed event time in the full data population:
\begin{align*}
    \theta = \E\{\nu(T^*)\},
\end{align*}
where $\nu$ is a known real-valued bounded transformation. For example, when $\nu(t) = \ind(t>t_0)$, $\theta$ is the marginal survival probability; when $\nu(t) = \min(t,t_0)$, $\theta$ is the restricted mean survival time. 
% \andrew{Tied back to real data example.} 
Both are commonly considered estimands in the time-to-event literature.
For the HAAS data example, practitioners may be interested in analyzing the cognitive impairment-free survival probabilities or the restrict mean survival time for time to cognitive impairment or death, whichever happens first.

We allow $Q^*$ and $T^*$ to be dependent, and
assume that an unmeasured latent factor  $U^*$, together with the measured covariates, explains the dependence between $Q^*$ and $T^*$. 
%Note that such a $U^*$ always exists as we can take $U^* = T^*$ or $U^* = Q^*$.
Suppose that the measured covariates can be classified into the three types of proxies as described in the introduction. 
%based on domain knowledge and denoted as $Z^*$, $W_1^*$, $W_2^*$ for types a, b and c, respectively.
Specifically,  a) $Z^*$ contains measured covariates that are directly associated with both $Q^*$ and $T^*$; b) $W_1^*$ contains the truncation proxies that may (but not necessarily) be associated with $Q^*$ but are only associated with $T^*$ through $(Z^*, U^*)$; and c) $W_2^*$ contains the event time proxies that may (but not necessarily) be associated with $T^*$ but are only associated with $Q^*$ through $(Z^*, U^*)$.
We formalize this in Assumption \ref{ass:proximal_indep} below.
\begin{assumption}[Proximal independence]\label{ass:proximal_indep}
    $(W_1^*,Q^*) \bigCI (W_2^*,T^*) \mid Z^*,U^*$. 
\end{assumption}
Assumption \ref{ass:proximal_indep} requires that $(Z^*,U^*)$ perfectly captures the dependence between $Q^*$ and $T^*$.
In addition, it requires that when conditioning on $(Z^*,U^*)$, the truncation-inducing proxy $W_1^*$ does not directly affect $T^*$ and the event-inducing proxy $W_2^*$ does not directly affect $Q^*$. 
The following Figure \ref{fig:DAG} shows an example DAG where Assumption \ref{ass:proximal_indep} holds.
% \andrew{Move exmaples here, add more plausibility}
\begin{figure}[h]
\centering
 \resizebox{180pt}{!}{%
   \begin{tikzpicture}[state/.style={circle, draw, minimum size=1.1cm}]
  \def\Ax{0}
  \def\Ay{0}
  \def\offset{2.5}
  \def\Bx{\Ax+5}
  \def\By{\Ay}
  \node[state,shape=circle,draw=black] (Z) at (\Bx-4,\Ay+1.5) {$W_1^*$};
  \node[state,shape=circle,draw=black] (Y) at (\Bx,\By) {$T^*$};
  \node[state,shape=circle,draw=black] (A) at (\Bx-2.5,\By) {$Q^*$};
  \node[state,shape=circle,draw=black] (X) at (\Bx-1.25,\By+1.5) {$Z^*$};
    \node[state,shape=circle,draw=black] (U) at (\Bx-1.25,\By+3.5) {$U^*$};
  \node[state,shape=circle,draw=black] (W) at (\Bx+1.5,\Ay+1.5) {$W_2^*$};

  \draw [dashed] (X) to [bend left=0] (W);
  %\draw [-latex] (A) to [bend left=0] (Y);
  \draw [-latex] (X) to [bend left=0] (A);
  \draw [dashed] (X) to [bend left=0] (Z);
  \draw [-latex] (X) to [bend left=0] (Y);
  \draw [-latex, dashed] (Z) to [bend left=0] (A);
 % \draw [latex-latex] (Z) to [bend left=-35] (A);
  \draw [-latex, dashed] (W) to [bend left=0] (Y);
 % \draw [latex-latex] (W) to [bend left=35] (Y);

  \draw [-latex] (U) to [bend left=0] (A);
  \draw [-latex] (U) to [bend left=0] (Z);
  \draw [-latex] (U) to [bend left=0] (Y);
  \draw [-latex] (U) to [bend left=0] (W);
  \draw [-latex] (U) to [bend left=0] (X);

\end{tikzpicture}
 }
\vskip -.1in
\caption{A DAG illustrating a situation where Assumption \ref{ass:proximal_indep} holds. Dashed lines denote causal relationships that may be present.}\label{fig:DAG}
\end{figure}
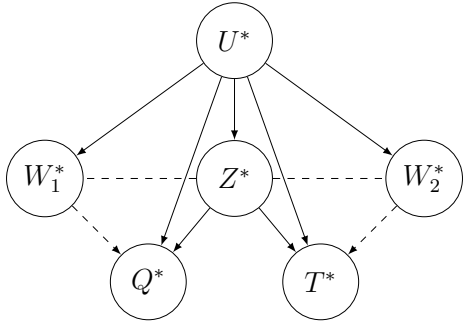

% \andrew{Remove to after Assumption \ref{ass:proximal_indep}. Make the intro shorter}
Below we give a concrete example for the three types of proxies using HAAS data.
As discussed in the introduction, latent factors such as socioeconomic status, health-seeking behavior, and overall health status are likely to be associated with both study entry age and age to cognitive impairment or mortality. 
Among the measured covariates, grip strength can be viewed as a proxy for overall health status. 
Since grip strength is unlikely to be the direct cause of cognitive impairment or death but rather associated with them through the latent overall health status, grip strength may be considered as a type `b' proxy.
On the other hand, education and lifestyle variables such as alcohol and cigarettes consumption are proxies for subjects' socioeconomic status and health-seeking behavior.
It is known that education and lifestyle variables may affect cognitive impairment and mortality, but they are unlikely to be direct associates with birth cohort or calendar time at study enrollment (and therefore left truncation time), except through the latent socioeconomic status and health-seeking behavior.
Therefore, they may be viewed as type `c' proxies. 
Other measured variables such as {\it APOE} genotype, systolic blood pressure, and heart rate may be considered as type `a' proxies. Figure \ref{fig:DAG_HAAS} illustrates the relationships among these variables.

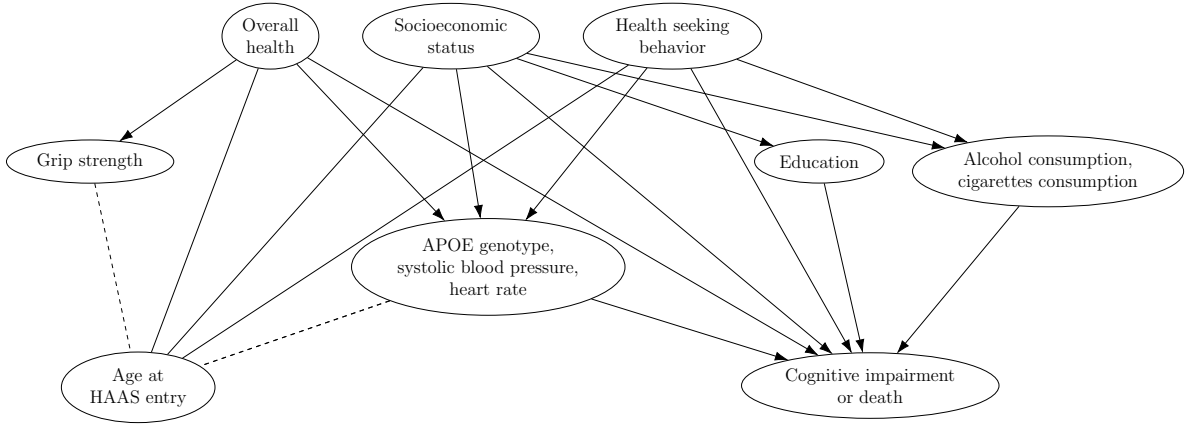
\begin{figure}[h]
\centering
 \resizebox{450pt}{!}{%
   \begin{tikzpicture}[
  node/.style={
    draw,
    ellipse,            % now recognized because of shapes.geometric
    align=center,
    minimum width=1.6cm,
    minimum height=1cm
  }
]
    \node[node] (UA) {Overall \\ health};
  % Nodes placed relative to A
  \node[node, right=of UA] (UB) {Socioeconomic \\ status};
  \node[node, right=of UB] (UC) {Health seeking \\ behavior};

  \node[node, below left=2cm and 2cm of UA] (WA) {Grip strength};
  \node[node, below right=2cm and 6cm of UB] (WB) {Education};
  \node[node, below right=2cm and 5cm of UC] (WC) {Alcohol consumption,\\ cigarettes consumption};
  \node[node, below right=4cm and 2cm of UA] (Z) {APOE genotype, \\ systolic blood pressure, \\  heart rate};

  \node[node, below left=7cm and 1cm  of UA] (Q) {Age at \\ HAAS entry};
  \node[node, below right=7cm and 1cm of UC] (T) {Cognitive impairment\\ or death};

   \draw[solid] (UA) to[bend left=0] (Q);
   \draw[-{Latex[length=3mm, width=2mm]}] (UA) to[bend left=0] (T);
   \draw[solid] (UB) to[bend left=0] (Q);
   \draw[-{Latex[length=3mm, width=2mm]}] (UB) to[bend left=0] (T);
   \draw[solid] (UC) to[bend left=0] (Q);
   \draw[-{Latex[length=3mm, width=2mm]}] (UC) to[bend left=0] (T);

  \draw[-{Latex[length=3mm, width=2mm]}] (UA) to[bend left=0] (WA);
  % \draw[dashed] (UA) to[bend left=0] (WB);
  % \draw[dashed] (UA) to[bend left=0] (WC);
  \draw[-{Latex[length=3mm, width=2mm]}] (UB) to[bend left=0] (WB);
  \draw[-{Latex[length=3mm, width=2mm]}] (UC) to[bend left=0] (WC);
  \draw[-{Latex[length=3mm, width=2mm]}] (UB) to[bend left=0] (WC);
  
  \draw[-{Latex[length=3mm, width=2mm]}] (WB) to[bend left=0] (T);
  \draw[-{Latex[length=3mm, width=2mm]}] (WC) to[bend left=0] (T);
  \draw[-{Latex[length=3mm, width=2mm]}] (Z) to[bend left=0] (T);
  \draw[dashed] (Z) to[bend left=0] (Q);
  \draw[dashed] (WA) to[bend left=0] (Q);
  \draw[dashed] (Z) to[bend left=0] (Q);

  \draw[-{Latex[length=3mm, width=2mm]}] (UA) to[bend left=0] (Z);
  \draw[-{Latex[length=3mm, width=2mm]}] (UB) to[bend left=0] (Z);
  \draw[-{Latex[length=3mm, width=2mm]}] (UC) to[bend left=0] (Z);
\end{tikzpicture}
 }
\vskip -.1in
\caption{A DAG illustration for HAAS data; dashed lines denote causal relationships that might be present; edges without arrows denote relationships that both directions are possible; the edges for relationships among the unmeasured latent factors and among the measured covariates are omitted for simplicity.}\label{fig:DAG_HAAS}
\end{figure}

We also make the following positivity assumption.

%Throughout the paper, we make the following assumptions for left truncation. 

\begin{assumption}[Latent positivity for truncation]\label{ass:positivity}
    $\PP(Q^*<t \mid Z^*=z, U^*=u) > 0$ for all $(t,z,u)$ in the support of $(T^*,Z^*,U^*)$. 
    % \lily{I don't mean the word latent but the assumption itself, it's not testable. We need to discuss this below if we need it}\yuyao{Added a discussion below}
    %\yuyao{I borrowed the word latent from \citet{ying2024proximal}. What about dropping the word "latent". Other proximal inference papers such as \citet{cui2024semiparametric} and \citet{tchetgen2024introduction} just name it as "positivity".}
\end{assumption}
Assumption \ref{ass:positivity} ensures that subjects across all levels of $(Z^*,U^*)$ have the opportunity to be observed.
It coincides with the positivity assumption in \citet{vakulenko2022nonparametric} when $U^*, W_1^*, W_2^*$ are empty sets.
% \andrew{plausibility}
In HAAS example, Assumption \ref{ass:positivity} requires that for any realization of the latent factors and type `a' proxies, the subject has a positivity chance of surviving long enough without cognitive impairment to enter the study.

We will use $\tauQ$ to denote the supremum for the support of $Q^*$, and
$G(t|z,u)$ to denote the conditional cumulative distribution function (CDF) of $Q^*$ given $(Z^*, U^*)$. 
We will use the notation $a\wedge c = \min(a,c)$ and $a\vee c = \max(a,c)$ for any real numbers $a$ and $c$.
% and $S_D(t) = \PP(D>t)$ to denote the survival function of $D$. 

% For the identification in Section \ref{sec:identification}, 

\section{Identification via truncation-inducing bridge process}\label{sec:identification}

% \lily{moved here}
To better illustrate the proposed framework for handling dependent left truncation, we will first focus on the setting without right censoring and then extend the framework to the more practical setting with right censoring. 

\subsection{Identification under no censoring}

In proximal causal inference and proximal survival analysis for handling dependent right censoring,  nonparametric identification is achieved via the so-called bridge functions, or bridge processes in  time-to-event settings, which are solutions of inverse problems that are formulated as integral equations leveraging the different types of proxies. 
% We borrow the identification ideas from 
% \andrew{delete below sentence} Motivated by the proximal causal inference and proximal survival analysis literature, we consider the truncation-inducing bridge process defined in Assumption \ref{ass:trunc_bridge} below.
Following the backwards counting process notation for $Q$ \citep{bickel1993efficient}, let
$\bar N_Q(t) = \ind(t\leq Q<T)$. We consider the following truncation-inducing bridge process.

\begin{assumption}[Existence of a truncation-inducing bridge process]\label{ass:trunc_bridge}
    There exists a bounded truncation-inducing bridge process $b(t,w_1,z)$ 
    % of pathwise bounded variation\yuyao{check if not needed} 
    satisfying %\yuyao{In AN proof we need $\E\{b(T,W_1,Z)^2\}<\infty$.}
    \begin{align}
        \E\{\dd b(t, W_1,Z) - \dd \bar N_Q(t) b(t,W_1,Z) \mid Q\leq t<T, W_2, Z\} = 0,  \label{eq:trunc_bridge_def}
    \end{align}
    with the initial condition 
    \begin{align}
        b(t, W_1,Z) = 1,\quad \text{for all  } t\geq \tauQ. \label{eq:trunc_bridge_initial_condi}
    \end{align}
\end{assumption}
% \andrew{Move some of these to appendix? currently too long.}

We note that \eqref{eq:trunc_bridge_def} in Assumption \ref{ass:trunc_bridge} is expressed in terms of the increments of stochastic processes $b(t,W_1,Z)$ and $\bar N_Q(t)$, and it is equivalent to 
\begin{align}
    \E\left[ \int_Q^T \varphi(t,W_2,Z) \{\dd b(t, W_1,Z) - \dd \bar N_Q(t) b(t,W_1,Z)\} \right] = 0  \label{eq:trunc_bridge_equation}
\end{align}
for any integrable function $\varphi(t,W_2,Z)$.

Equation \eqref{eq:trunc_bridge_def} resembles the equation that defines the censoring-inducing bridge process in \citet{ying2024proximal} for handling dependent right censoring 
(see Appendix \ref{app:plausibility} for a detailed comparison).
Similar equations have also been used in proximal causal inference to define bridge functions, 
% which are Fredholm integral equations, and 
for which the existence of a solution is given by Picard's theorem under certain regularity conditions \citep{miao2018identifying, ying2023proximal, cui2024semiparametric}. 
Unlike the bridge functions in proximal causal inference, the bridge process in this work (and in \citet{ying2024proximal}) is a stochastic process that involves the time dimension. 
Investigating the existence of such a process for left truncated data would be of interest in probability theory but is beyond the scope of the current paper. As in \citet{ying2024proximal}, we take the existence of the truncation-inducing bridge process as a primitive assumption.
% Investigating the existence of such a process for left truncated data would be of interest in probability theory but is beyond the scope of the current paper. 

% We include in Appendix \ref{app:plausibility} detailed discussions of Assumption \ref{ass:trunc_bridge} and related assumptions in the proximal inference literature. 
To assure the reader, we provide in Appendix \ref{app:plausibility} a concrete parametric data generating mechanism where such a bridge process exists; note that, however, our framework does not rely on this parametric data generating mechanism.

\begin{assumption}[Completeness] \label{ass:completeness_truncbridge}
    For any $t>0$ and any integrable function $\zeta(t,Z,U)$, 
    $
    \E\left[\zeta(t,Z,U) \mid Q\leq t<T, W_2,Z \right] = 0
    $
    if and only if $\zeta(t,Z,U) = 0$ a.s..
\end{assumption}
Assumption \ref{ass:completeness_truncbridge} requires $W_2$ to be relevant to $U$ and that there is enough variability in $W_2$ compared to the variability in $U$. 
Note that 
$ \E\left[\zeta(t,Z,U) \mid Q\leq t<T, W_2, Z \right] = $ \\ $\E\left[\zeta(t,Z^*,U^*) \mid Q^*\leq t<T^*, W_2^*, Z^* \right]$, 
so Assumption \ref{ass:completeness_truncbridge} excludes the independence of $U^*$ and $W_2^*$ conditional on $(Z^*, Q^*\leq t<T^*)$ for any $t>0$. 
For the case where $W_2$ and $U$ are categorical variables, Assumption \ref{ass:completeness_truncbridge} requires that the number of categories of $W_2$ to be at least as many as that of $U$. 
Similar completeness assumptions were also considered in the proximal inference literature for treatment effect estimation and for handling right censoring; see for example \citet{tchetgen2024introduction}, \citet{cui2024semiparametric}, \citet{ying2023proximal}, \citet{ying2024proximal} and the references therein for great discussions. 
% \andrew{Add plausibility of assumption by HAAS example}
For the HAAS example, Assumption \ref{ass:completeness_truncbridge} requires that the education and lifestyle variables have sufficient variability and contain enough information about the latent socioeconomic status, health-seeking behavior, and overall health status in each at-risk subpopulation (i.e., those who have entered the study but not yet experienced the event) at time $t$, for all $t$. 

With the above assumptions, we have the following 
% proximal truncation-inducing 
identification in the case with no censoring.

\begin{lemma}[Proximal truncation-inducing identification]\label{thm:proximal_trunc_identification}
    Under Assumptions \ref{ass:proximal_indep} -  \ref{ass:completeness_truncbridge}, for any truncation-inducing bridge process $\{b(t,W_1,Z): t\geq 0\}$ satisfying \eqref{eq:trunc_bridge_def} and \eqref{eq:trunc_bridge_initial_condi}, we have 
    \begin{align}
        \E\{\dd b(t, W_1,Z) - \dd \bar N_Q(t) b(t,W_1,Z) \mid Q\leq t<T, Z,U\} = 0.  \label{eq:trunc_bridge_UZ}
    \end{align}
    Furthermore, 
    \begin{align}
        \theta = \frac{\E\{b(T,W_1,Z) \nu(T)\}}{\E\{b(T,W_1,Z)\}} \label{eq:identification}
    \end{align}
\end{lemma}

The proof of Lemma \ref{thm:proximal_trunc_identification} is in Appendix \ref{app:proof_identification}.
We note that the identification does not require uniqueness of the truncation-inducing bridge process that satisfies \eqref{eq:trunc_bridge_def} and \eqref{eq:trunc_bridge_initial_condi}. Any process satisfying \eqref{eq:trunc_bridge_def} and \eqref{eq:trunc_bridge_initial_condi} can be used to identify $\theta$.

\begin{remark} \label{rm:IPQW}
    When the conditional independent left truncation holds with $Q^*\bigCI T^*\mid Z^*$, corresponding to the case of $U^* = \varnothing$, $W_1^* = \varnothing$, and $W_2^* = \varnothing$ in our setting, $G$ denotes the conditional CDF of $Q^*$ given $Z^*$. 
    In this case, the inverse probability of truncation weight $b(t,Z) = 1/G(t|Z)$ satisfies \eqref{eq:trunc_bridge_def} (see Appendix \ref{app:special_case} for a proof).
    With this $b(t,Z)$, the identification in Lemma \ref{thm:proximal_trunc_identification} becomes the inverse probability of truncation weighting identification \citep{vakulenko2022nonparametric, wang2024doubly}: 
    $$\theta = \frac{\E\{\nu(T)/G(T|Z)\}}{\E\{1/G(T|Z)\}}.$$
    % $$\theta = \frac{\E\{\nu(T)/G(T|Z)\}}{\E\{1/G(T|Z)\}}.$$
\end{remark}

% Lemma \ref{thm:proximal_trunc_identification}  motivates a weighting estimator for $\theta$, which is described in the next section.

\subsection{Identification with right censoring}

%Besides left truncation, the time-to-event outcome is also commonly subject to right censoring, where the event time is not always observed. 
We consider possible loss to follow-up after study entry, and assume
 that $C^*>Q^*$ almost surely. Let $D^* = C^* - Q^*$ denote the residual censoring time.
For each observed subject, we observe $O = (Q,X,\Delta,W_1,W_2,Z)$, where $X = \min(T,C)$ and $\Delta = \ind(T<C)$.
For $D = C - Q$, let $S_D(t) = \PP(D>t)$. 
% \lily{say something about other censoring scenario, maybe refer to discussion.}
Besides this censoring scenario, another possible scenario is that censoring may happen before truncation \citep{qian2014assumptions}, which we discuss in Section \ref{sec:discussion}. Here we focus on the setting where censoring is always after truncation since this is the case for the HAAS data.

%\lily{discuss $ t_0 > \tau_{max}$}
In the presence of right censoring, the distribution of $T^*$ is known to be non-identifiable after the maximum follow up time $\tauC$, so any reasonable estimand should be a functional of the distribution of $T^*$ up to $\tauC$. For examples of the marginal survival probabilities and RMST considered in Section \ref{sec:preliminary}, this means that $t_0$ should not exceed $\tauC$.

We make the following assumptions about right censoring. 
\begin{assumption}[Noninformative residual censoring]\label{ass:cen_random}
    $D \bigCI (Q,T,W_1,W_2,Z)$. 
\end{assumption}

\begin{assumption}[Positivity for censoring]\label{ass:cen_positivity}
% \yuyaox{Let $\tnu\in (0,\infty]$ denote the constant such that $\nu(t) = \nu(\tnu)$ for all $t\geq \tnu$, and denote $\tmax = \tnu \vee \tauQ$.}
% There exists $\etaD>0$ such that $S_D(T\wedge\tmax - Q) > \etaD$ almost surely. 
There exists $\etaD>0$ such that $S_D(T - Q) > \etaD$ almost surely. 
% \blue{$S_D(\tilde T - Q) > \etaD$ almost surely.}
\end{assumption}

Assumption \ref{ass:cen_random} {is} an independence assumption on the residual censoring time. Such assumptions are also considered in \citet{vakulenko2022nonparametric} and \citet{wang2024doubly}.
Assumption \ref{ass:cen_positivity} is a commonly considered positivity assumption for time-to-event data under right censoring, which ensures that the parameter of interest can be identified. 

We adapt the inverse probability of censoring weighting (IPCW) \citep{robins1992recovery} to the residual time scale, leading to the following identification.

\begin{lemma}[Proximal truncation-inducing identification under right censoring]\label{thm:identification_cen_after}
    Under Assumptions \ref{ass:proximal_indep} - 
    % \ref{ass:completeness_truncbridge}, \ref{ass:cen_random}, and 
    \ref{ass:cen_positivity}, for any truncation-inducing bridge process $b$ satisfying \eqref{eq:trunc_bridge_def} and \eqref{eq:trunc_bridge_initial_condi}, we have 
    % \lily{how different is the two implementation?}\yuyao{Only different at the last step.}
    \begin{align}
        \theta =  \left. \E\left\{ \frac{\Delta b(X,W_1,Z) \nu(X) }{S_D(X - Q)} \right\} \right/
        \E\left\{\frac{\Delta b(X,W_1,Z)}{ S_D(X - Q)} \right\}. \label{eq:identification_C}
    \end{align}
    % \blue{
    % \begin{align}
    %     \theta =  \left. \E\left\{ \frac{\tilde\Delta b(\tilde X,W_1,Z) \nu(\tilde X) }{S_D(\tilde X - Q)} \right\} \right/
    %     \E\left\{\frac{\tilde\Delta b(\tilde X,W_1,Z)}{ S_D(\tilde X - Q)} \right\}.
    % \end{align}
    % }
\end{lemma}

The proof of Lemma \ref{thm:identification_cen_after} is in Appendix \ref{app:identification_proof_cen_after}.

\section{Proximal estimation and inference}\label{sec:estimation_asymptotics}

% \andrew{Combine. Keep some in appendix. Move the no censoring setting. Just say insome sentences: incorporate IPCW weights to handle censoring. }

\subsection{Estimation}\label{sec:estimation}

The identification in Lemma \ref{thm:identification_cen_after} involves two unknown nuisance parameters: the truncation-inducing bridge process $b$, and the survival function $S_D$ for the residual censoring time.  
With a random sample $\{O_i\}_{i=1}^n$,
if $b$ and $S_D$ are known, a natural estimator for $\theta$ motivated from Lemma \ref{thm:identification_cen_after} is
\begin{align}
    \left. \left\{\sum_{i=1}^n \frac{\Delta_i b(X_i,W_{1i},Z_i)\nu(X_i)}{S_D(X_i-Q_i)} \right\} \right/ \left\{\sum_{i=1}^n \frac{\Delta_i b(X_i,W_{1i},Z_i)}{S_D(X_i-Q_i)} \right\}. \label{eq:est_oracle}
\end{align}

In practice, $b$ and $S_D$ are unknown. If we have their estimators $\hat b$ and $\hat S_D$, substituting them into \eqref{eq:est_oracle} results in an estimator for $\theta$:  
\begin{align}
    \httc
    = \left.\left\{\sum_{i=1}^n \frac{\Delta_i \hat b(X_i,W_{1i},Z_i)\nu(X_i)}{\hat S_D(X_i - Q_i)} \right\} \right/ \left\{\sum_{i=1}^n \frac{\Delta_i \hat b(X_i,W_{1i},Z_i)}{\hat S_D(X_i - Q_i)} \right\}.  \label{eq:theta_hat_c}
\end{align}

% We first consider estimating $S_D$.
The $S_D$ can be estimated by the Kaplan-Meier estimator using data on the residual time scale, i.e., time since study entry, and we denote the estimator by $\hat S_D$. This is because $D$ is right censored by $T-Q$ in the observed data and we have $D\bigCI (T-Q)$ by Assumption \ref{ass:cen_random}. 

For estimating $b$, the integral equations in \eqref{eq:trunc_bridge_def} are known to be ill-posed \citep{ai2003efficient}, which means that small uncertainties in estimating the conditional expectations can lead to large uncertainties in the solution.
Therefore, directly solving the integral equations will lead to unstable estimates.

There have been multiple strategies developed in the proximal inference literature to overcome this challenge. 
One approach is to impose compatible parametric or semiparametric working models for the bridge functions or processes  \citep{tchetgen2024introduction, cui2024semiparametric, ying2023proximal, ying2024proximal}, which can be viewed as a form of regularization for stability \citep{tchetgen2024introduction}. 
Alternatively, nonparametric methods have recently been developed for estimating the bridge functions in proximal identification of the average treatment effect, including adversarial learning \citep{ghassami2022minimax, kallus2021causal, olivas2025source} and debiased ill-posed regression \citep{ghassami2025debiased}.

As an illustration, we consider the following semiparametric working model for $b$ with the additive form: 
\begin{align}
    b(t,W_1,Z;B(t)) = \exp\{B_0(t) + W_1 B_1(t) + Z B_z(t) \}, \label{eq:model_b}
\end{align}
where $B_0(t)$, $B_1(t)$, and $B_z(t)$ are arbitrary bounded functions of $t \in[0,\tauQ]$ with the initial condition $B_0(\tauQ) =  B_1(\tauQ) = B_z(\tauQ) = 0$;
when the variables $W_1$ and $Z$ are multi-dimensional, we treat $W_1$ and $Z$ as row vectors and take $B_1(t)$ and $B_z(t)$ as column vectors for each fixed $t$.
The condition $B_0(\tauQ) =  B_1(\tauQ) = B_z(\tauQ) = 0$ ensures that the initial condition in \eqref{eq:trunc_bridge_initial_condi} is satisfied. 
In Appendix \ref{app:model_compat}, we show that \eqref{eq:model_b} is a compatible working model by providing an example of a data generating mechanism that satisfies Assumptions \ref{ass:proximal_indep} - \ref{ass:completeness_truncbridge}, under which there exists a bridge process satisfying model \eqref{eq:model_b} and conditions \eqref{eq:trunc_bridge_def} and \eqref{eq:trunc_bridge_initial_condi}.
Note that, however, our identification and estimation framework is general and does not rely on this specific data generating mechanism.

% \yuyao{Still include the following two equations for the case without right censoring here because the time varying IPCW weights are motivated by \eqref{eq:trunc_bridge_EE_d}. Keeping the below may help the reader better understand why the weights $\Delta(t)/S_D(X\wedge t-Q)$ are used for handling right censoring.}
Under model \eqref{eq:model_b}, if the event time $T$ were observed, one can estimate $b$ by solving the empirical version of \eqref{eq:trunc_bridge_equation} with a properly chosen function $\varphi(t,W_2,Z)$.
We take $\varphi(t,W_2,Z) = (1,W_{2},Z)^\top$, resulting in the following estimating equation for $b$ in the censoring-free data: 
\begin{align}
    \frac{1}{n} \sum_{i=1}^n \int_{Q_i}^{T_i} (1,W_{2i},Z_i)^\top  \{\dd b(t, W_{1i},Z_i) - \dd \bar N_{Qi}(t) b(t,W_{1i},Z_i)\} = 0. \label{eq:trunc_bridge_EE_int}
\end{align}
Let $B(t) = (B_0(t), B_1(t)^\top, B_z(t)^\top)^\top$.
By plugging in the semiparametric model \eqref{eq:model_b} into \eqref{eq:trunc_bridge_EE_int}, 
it suffice to solve the following set of estimating equations for $B(t)$ in the differential form: 
for all $t\in[0,\tauQ)$,
\begin{align}
    & \frac{1}{n} \sum_{i=1}^n \ind(Q_i\leq t<T_i) 
    \exp\{(1,W_{1i},Z_i)B({t+}) \} \nonumber \\
    &\qquad\qquad \cdot (1,W_{2i},Z_i)^\top  \{(1,W_{1i},Z_i) dB(t)  - \dd \bar N_{Qi}(t)\} = 0.  \label{eq:trunc_bridge_EE_d}
\end{align} 

Under right censoring, the estimating equations in \eqref{eq:trunc_bridge_EE_d} does not directly apply because it involves $\ind(Q\leq t < T)$, which is not always observed due to right censoring. 
To handle this missingness, we again apply IPCW. 
Specifically, let $\Delta(t) = \ind(\{\Delta = 1\}\cup\{\Delta = 0,t<X\})$.
Since $\ind(Q\leq t < T) = \ind(Q\leq t < X)$ if $T\vee t < C$ (i.e., if either $\Delta = 1$ or $t<X$), 
by incorporating time-varying IPCW weights to \eqref{eq:trunc_bridge_EE_d}, 
we have the following set of estimating equations for $B(t)$: for $t\in[0,\tau)$, 
\begin{align}
    & \sum_{i=1}^n \frac{\Delta_i(t)}{\hat S_D(X_i\wedge t - Q_i)} \cdot \ind(Q_i\leq t<X_i) 
    \exp\{(1,W_{1i},Z_i)B({t+})\} \nonumber \\
    &\qquad\qquad \cdot (1,W_{2i},Z_i)^\top \{(1,W_{1i},Z_i) dB(t) - \dd \bar N_{Qi}(t)\} =  0.   \label{eq:trunc_bridge_EE_d_cen}
\end{align}

Therefore, starting from the initial values $\hat B(\tauQ) = 0$, we can compute $\hat B(t)$ backwards in time: 
\begin{align}
    \hat B(t) = - \frac{1}{n} \int_t^{\tauQ} \M_B(t)^\dagger \ \sum_{i=1}^n  \frac{\Delta_i(t)\ind(Q_i\leq t<T_i)}{\hat S_D(X_i\wedge t - Q_i)} 
    \cdot \exp\{(1, W_{1i}, Z_i) \hat B({t+}) \} (1,W_{2i},Z_i)^\top \dd \bar N_{Qi}(t), 
\end{align}
where $\dagger$ denote the Moore-Penrose inverse (also called pseudo-inverse) of a matrix, and 
\begin{align}
    \M_B(t) =  \frac{1}{n}\sum_{i=1}^n \frac{\Delta_i(t)\ind(Q_i\leq t<T_i)}{\hat S_D(X_i\wedge t - Q_i)} \cdot \exp\{(1, W_{1i}, Z_i) \hat B({t+}) \} (1,W_{2i},Z_i)^\top(1,W_{1i},Z_i). 
\end{align}
Since the $\bar N_{Qi}(t)$'s are left continuous step functions that only jumps at $Q_i$'s, so is $\hat B(t)$.

Finally, we take $\hat b(t,w_1,z) = b(t,w_1,z; \hat B(t))$ in \eqref{eq:theta_hat_c} to obtain the estimator $\httc$.

Note that $\httc$ incorporates inverse probability of censoring weights $\Delta/\hat S_D(X-Q)$.
For certain choices of $\nu$, such as $\nu(t) = \ind(t>t_0)$ or $\nu(t) = \min(t,t_0)$, the estimator $\httc$ can be improved in terms of stability by adjusting the inverse probability of censoring weights using the minimum time at which both $\nu(T)$ and $b(T,W_1,Z)$ are observed \citep{robins1992recovery}. Details of the adjusted estimator are provided in Appendix \ref{app:theta_hat_c_adjusted}. This adjustment is implemented in both the simulation and data application below.

% Besides the time-varying IPCW weights used in  \eqref{eq:trunc_bridge_EE_d_cen}, another natural idea is to attach case weights $\Delta/S_D(X-Q)$ for each subject to handle right censoring. 
% % This way the weighted uncensored subjects can be viewed as a random sample from the truncated and uncensored population, so the identification and estimation in Sections \ref{sec:identification} and \ref{sec:estimation} can be applied to the weighted sample.
% However, this approach requires a stronger positivity assumption, $S_D(T-Q)$ being bounded away from zero almost surely, which is less likely to hold in practice. 
% In addition, the case weights $\Delta/\hat S_D(X-Q)$ tend to be more extreme than the time-varying weights $\Delta(\tmax)/\hat S_D(X\wedge\tmax-Q)$ and $\Delta(t)/\hat S_D(X\wedge t -Q)$ used in \eqref{eq:trunc_bridge_EE_d_cen} and \eqref{eq:theta_hat_c}, resulting in more instability for the estimator of $\theta$. The numerical results in the later simulation section supports this intuition.

\subsection{Asymptotics}

The estimator $\httc$ is consistent and asymptotically normal under the following generic assumptions on $\hat b$.
We first introduce some norm notation. For a random variable $Y$ with support $\mathcal{Y}$, denote the norms $\|Y\|_1 = \E(|Y|)$ and $\|Y\|_{\sup} = \sup_{y\in\mathcal{Y}}|y|$.
We consider the following assumptions for $\hat b$.
% Let $\btr$ denote the bridge process that $\hat b$ converges to.

\begin{assumption}[Consistency] \label{ass:consistency}
$\|\hat b(T,W_1,Z)- \btr(T,W_1,Z)\|_{1} = o(1)$ for some stochastic process $\btr$. 
\end{assumption}

\begin{assumption}[Asymptotic linearity] \label{ass:AL}
The estimator 
$\hat b(t,w_1,z)$ is asymptotically linear with influence function $\xi(t, w_1, z; O)$. 
% such that $\E\{\xi(t,w_1,z; O)\} = 0$ and $\E\{\xi(t,w_1,z; O)^2\} < \infty$ for all $(t, w_1, z)$ in the support of $(T,W_1,Z)$, and
In addition, denote 
\begin{align*}
    R(t,w_1,z) = \hat b(t,w_1,z) - \btr(t,w_1,z) - \frac{1}{n}\sum_{i=1}^n \xi(t, w_1, z; O_i); 
\end{align*}
suppose that $\|R(T,W_1,Z)\|_1  = o(n^{-1/2})$ and $\E\left\{\xi(T_i,W_{1i},Z_i; O_j)^2 \right\} < \infty$ for $i\neq j$.
\end{assumption}

Assumption \ref{ass:consistency} requires that $\hat b(T,W_1,Z)$ is a consistent estimator for $\btr(T,W_1,Z)$, which is a reasonable assumption when the model class used to estimate $b$ contains the truth.  
Assumption \ref{ass:AL} requires that $\hat b(T, W_1,Z)$ is an asymptotically linear estimator for $\btr(T,W_1,Z)$, which usually holds when parametric or semiparametric models are used. 
It can be shown that Assumptions \ref{ass:consistency} and \ref{ass:AL} hold for the $\hat b$ obtained in Section \ref{sec:estimation} when the semiparametric model \eqref{eq:model_b} is correctly specified. The proof follows a similar argument as that in \citet[Section C of the Supplementary Material]{ying2024proximal}, using the fact that the Kaplan-Meier estimator $\hat S_D$ is uniformly consistent and asymptotically linear. 
% , with a similar argument as in Section \ref{sec:estimation:inference}. 
% When this is the case, Theorem \ref{thm:AN_c} implies that the estimator is consistent and asymptotically normal, and its asymptotic variance and confidence interval can be obtained via the random weighting bootstrap described in Section \ref{sec:estimation:inference}. 

The estimator $\httc$ is consistent and asymptotically normal under the above generic assumptions on $\hat b$, which is formally stated in the theorem below.

\begin{theorem}\label{thm:AN_c}
    Under Assumptions \ref{ass:proximal_indep} - 
    % \ref{ass:completeness_truncbridge},  \ref{ass:cen_random}, and 
    \ref{ass:cen_positivity}, 
    % - \ref{ass:consistency_cen},
    suppose that $\hat b$ satisfies Assumption \ref{ass:consistency}, \\ $\|\hat b(T,W_1,Z)\|_{\sup}<\infty$, and that $\hat b(t,W_1,Z) = \hat b(\tauQ,W_1,Z)$ for all $t\geq \tauQ$. 
    Then 
    (i) $\httc\convp \thetatr$ as $n\to\infty$;
    (ii) if in addition $\hat b$ satisfies Assumption \ref{ass:AL},  $\httc$ is asymptotically linear and 
    $\sqrt{n}(\httc - \thetatr) \convd N(0, \sigmac^2)$, 
    where $\sigmac^2>0$ depend on $\xi$.
\end{theorem}

The proof is in Appendix \ref{app:proof_AN_withC}. 
The asymptotic variance $\sigmac^2$ is affected by the influence function of $\hat b$, which depends on the model used to estimate $b$. 
Nevertheless, bootstrap can be applied to estimate the asymptotic variance since $\httc$ is asymptotically linear. 
We apply the random weighting bootstrap \citep[Chapter 10]{shao2012jackknife}, also known as the bootstrap clone method \citep{lo1991bayesian} or the multiplier bootstrap \citep{kosorok2008introduction}, which is closely related to the Bayesian bootstrap \citep{rubin1981bayesian}.
Specifically, for each bootstrap replication $r = 1,..., R$, we generate $e_{1,r},...,e_{n,r}$ i.i.d. from exponential distribution with rate parameter 1.  We then get the bootstrap weights $w_{i,r} = e_{i,r}/\sum_{i=1}^n e_{i,r}$ for $i = 1, ... ,n$. The weights are used in both the estimation of the nuisance parameters $b$ and $S_D$ and the final estimation of $\theta$;  we denote the obtained estimator by  $\httc_{r}^*$. The bootstrapped variance is the empirical variance of $\httc_{1}^*, ..., \httc_{R}^*$, and the Wald-type confidence interval is constructed based on the bootstrapped variance.  
Compared with the resampling-based bootstrap methods, the random weighting bootstrap avoid the issue of creating ties for time-to-even data, which is preferable.

\section{Simulation}\label{sec:simu}

We generate $Z^*,W_1^*,W_2^*, T^*$ from the following distributions   \citep{ying2024proximal}: \\
$Z^* \sim \max\{N(0.6, 0.45^2), 0\}$, $U^* \sim \max\{N(0.6, 0.45^2), 0\}$, $W_1^* \sim N(1.4 + 0.3Z^* - 0.9U^*, 0.25^2)$, $W_2^* \sim N(0.6 - 0.2Z^* + 0.5U, 0.25^2)$, $\PP(T^*>t \mid Z^*,U^*) = \exp\{ - (0.25 + 0.3Z^*+0.6U^*)t\}$. 
The left truncation time $Q^*$ is generated from %the distribution:
\begin{align*}
    \PP(\tauQ - Q^* > t \mid Z^*,U^*) 
    = \left\{
    \begin{array}{ll}
        \exp\{ - (0.1 + 0.25Z^* + U^*)t\}, &  t <  \tauQ; \\
        0, &  t\geq \tauQ. \\
    \end{array}
    \right.
\end{align*}
The residual censoring time $D^*$ is generated from Weibull distribution with shape parameter 2 and scale parameter 2.
Then $C^* = Q^* + D^*$.
Subjects are observed only if $Q^*<T^*$.
The resulting truncation rate is around 47\%. The resulting censoring rate is around 37\% in the observed data.
% Wed Appendix \ref{supp:simu_noC} includes additional simulation results in the same data generating mechanism but no censoring. 

We consider the estimand $\theta = \PP(T^*>1) = 0.4632$, corresponding to $\nu(t) = \ind(t>1)$, which is computed from a simulated full data sample of size $10^5$.
We consider the estimator $\hat\theta$ with adjusted IPCW weights (see Appendix \ref{app:theta_hat_c_adjusted} for details), denoted by ``PQB" (which stands for \textbf{P}roximal Truncation(\textbf{Q})-inducing \textbf{B}ridge estimator) in the tables.
% Recall the notation $\tilde\Delta$ and $\tilde X$ from Remark \ref{remark:IPCW}.  
Let $\tilde X = \min(X,t_0 \vee \tauQ)$, and $\tilde\Delta = \ind(\{\Delta = 1\} \cup \{\Delta = 0, t_0 \vee \tauQ<X\})$.
We also consider the inverse probability of truncation weighting estimator (IPQW) with the adjusted IPCW weights to handle right censoring:
\begin{align*}
    \hat\theta_{\text{IPQW}} = \left.\left\{ \sum_{i=1}^n \frac{\tilde \Delta_i \nu(\tilde X_i)}{\hat S_D(\tilde X_i-Q_i)\hat H(\tilde X_i|W_{1i},W_{2i}, Z_i)} \right\} \right/ \left\{\sum_{i=1}^n \frac{\tilde\Delta_i}{\hat S_D(\tilde X_i-Q_i)\hat H(\tilde X_i|W_{1i},W_{2i}, Z_i)} \right\},
\end{align*}
% ``IPQW" and ``IPQW-U" denote the inverse probability weighting estimators that uses weights $\tilde\Delta/\{\hat S_{D}(\tilde X - Q) \hat H(\tilde X|W_{1},W_{2}, Z)\}$ or  $\tilde\Delta/\{\hat S_{D}(\tilde X - Q) \hat G(\tilde X|Z,U)\}$, respectively; recall that 
where $H$ denotes the conditional CDF of $Q^*$ given $(W_1^*, W_2^*, Z^*)$,
and $\hat H$ is obtained by fitting an Aalen model for $Q^*$ given $(W_1^*,W_2^*,Z^*)$ on the reversed time scale: 
$
H(t|w_1,w_2,z) = \exp\{\alpha_0(t) + \alpha_1(t)w_1 + \alpha_2(t)w_2 + \alpha_3(t)z\}
$. 
%where $\alpha_0(t), \alpha_1(t), \alpha_2(t), \alpha_3(t)$ are the unknown parameters.
Specifically, we compute $\hat H$ by solving the coefficients $\alpha_0(t), \alpha_1(t), \alpha_2(t), \alpha_3(t)$ backwards in time starting with the initial condition that $\alpha_0(\tauQ)=  \alpha_1(\tauQ)= \alpha_2(\tauQ)= \alpha_3(\tauQ) = 0$, using the same procedure as for estimating $b$ in Section \ref{sec:estimation}. This approach is computationally faster than  using the \texttt{aalen()} function in the R package \texttt{timereg}. The IPQW estimator requires $Q^*\bigCI T^*\mid (W_1^*,W_2^*, Z^*)$, which is violated under this data generating mechanism due to the unmeasured dependence-inducing variable $U^*$. 

We also consider the product-limit (PL) estimator \citep{wang1991nonparametric} that random left truncation and right censoring, the Kaplan-Meier (KM) estimator that ignores left truncation, and the naive estimator that ignores both left truncation and right censoring and simply averages the $\nu(X_i)$'s.
% Recall that $G$ denotes the conditional CDF of $Q^*$ given $(Z^*,U^*)$.
As benchmarks, we consider two inverse probability weighting ``estimators" that incorporate $(Z,U)$ in the weights: the ``IPQW-U"  uses weights $\tilde\Delta/\{\hat S_D(X-Q)\hat G(T|Z,U)\}$ with $\hat G$ estimated by fitting an {Aalen} model for $Q^*$ given $(Z^*,U^*)$ on the reversed time scale; and the ``IPQW-o" is an oracle inverse probability of truncation weighting estimator that uses the true weights $\tilde\Delta/\{S_D(X-Q)G(X|Z,U)\}$.
Note that these two ``estimators" are not feasible in practice because they rely on the unobserved variable $U$; we include them in the simulation only for comparison purposes.

To study the impact of using time-varying IPCW weights versus using the case weights $\Delta/S_D(X-Q)$, we additionally consider the PQB, IPQW, and IPQW-U estimators that instead incorporate the IPCW case weights in both the estimation of $b$, $G$ and $H$, and the final estimation for $\theta$. These estimators are denoted ``PQB-cw", ``IPQW-cw", and ``IPQW-U-cw", respectively.

We report the bias, the empirical standard deviation (SD), the mean of the bootstrapped standard errors (bootSE), and the coverage probability (CP) of the 95\% Wald-type confidence intervals based on the bootstrapped standard errors. For all estimators, the bootstrapped standard errors are obtained using random weighting bootstrap with 200  replications. 

The results are summarized in Table \ref{tab:simu2_cens}.
As expected, the IPQW-U and the IPQW-o estimators have small bias and close to nominal coverage for sample size 1000.
Our PQB estimator also has a small bias and close to nominal coverage when the sample size is 1000, while the confidence intervals {are} slightly undercovered when {the} sample size is 500. 
We also observe that the PQB estimator shows a larger SD and bootSE compared with the other estimators, reflecting the difficulty of solving the bridge process.
The IPQW estimator shows a large bias, likely due to the unobserved dependence inducing variable $U$.
In addition, we observe that the ``PQB-cw", ``IPQW-cw", and ``IPQW-U-cw" estimators show large biases and low coverages, likely due to violation Assumption \ref{ass:cen_positivity} (positivity for censoring), as the support of $T^*$ is $[0,\infty)$.
% \lily{elaborate please} \yuyaox{Assumption \ref{ass:cen_positivity} since the support of $T^*$ is $[0,\infty)$.}
Finally we see that the PL estimator shows a large bias and  low coverage, and the KM estimator significantly overestimates the survival probability. 
%consistent with the direction of the left truncation bias. 
The naive estimator shows a positive bias, whose magnitude is smaller than that of the KM estimator. This is because the bias directions for left truncation and right censoring are opposite:  ignoring left truncation tends to overestimate the time‑to‑event, and ignoring right censoring tends to underestimate it.
% \lily{how do you know}

\begin{table}[ht]
\centering
\caption{Simulation results for different estimators under the case with right censoring. Each observed data set has sample size 500 or 1000, and 500 data sets are simulated for each sample size.
% Aalen models are fitted on the reversed time scale to estimate $G$ for the IPQW estimators with the following covariates: $(Z, W_1, W_2)$ for IPQW; $(Z, W_1)$ for IPQW2; $(Z,U)$ for IPQWu. Estimators without `-tv' denote that the IPCW weights $\Delta/S_D(X-Q)$ are used for both the estimation of $B(t)$ and the estimation of $\theta$; `-cw' denote that the IPCW weights $\Delta/S_D(X-Q)$ are used as case weights to estimate $B(t)$ and the IPCW weights with adjustment are used for the estimation of $\theta$.
}
\label{tab:simu2_cens}
\begin{tabular}{lrcccrccc}
  \toprule
  & \multicolumn{4}{c}{$n = 500$} & \multicolumn{4}{c}{$n = 1000$} \\
  \cmidrule(lr){2-5} \cmidrule(lr){6-9}
  Method & Bias & SD & bootSE & CP & Bias & SD & bootSE & CP \\
  \midrule
  PQB & -0.0124 & 0.0745 & 0.0605 & 0.926 & -0.0062 & 0.0395 & 0.0393 & 0.944 \\
  IPQW & 0.0168 & 0.0426 & 0.0391 & 0.906 & 0.0172 & 0.0284 & 0.0278 & 0.882 \\
  % % IPQW2 & 0.0225 & 0.0411 & 0.0381 & 0.890 & 0.0219 & 0.0272 & 0.0272 & 0.844 \\
  PQB-cw & -0.0209 & 0.0858 & 0.0710 & 0.928 & -0.0565 & 0.0623 & 0.0558 & 0.872 \\
  IPQW-cw & 0.0116 & 0.0479 & 0.0435 & 0.906 & -0.0325 & 0.0379 & 0.0355 & 0.832 \\
  % IPQW2-cw & 0.0189 & 0.0450 & 0.0421 & 0.904 & -0.0263 & 0.0362 & 0.0350 & 0.866 \\
  PL & 0.0710 & 0.0312 & 0.0293 & 0.324 & 0.0707 & 0.0208 & 0.0207 & 0.072 \\ 
  KM & 0.2633 & 0.0222 & 0.0203 & 0.000 & 0.2631 & 0.0145 & 0.0144 & 0.000 \\  
  naive & 0.1889 & 0.0222 & 0.0212 & 0.000 & 0.1891 & 0.0157 & 0.0150 & 0.000 \\
  \midrule
  IPQW-U & -0.0022 & 0.0479 & 0.0426 & 0.922 & -0.0024 & 0.0326 & 0.0307 & 0.940 \\
  IPQW-U-cw & -0.0065 & 0.0539 & 0.0471 & 0.920 & -0.0466 & 0.0419 & 0.0380 & 0.762 \\
  % PQB-noC & -0.0054 & 0.0551 & 0.0567 & 0.930 & -0.0055 & 0.0373 & 0.0371 & 0.948 \\
  IPQW-o & 0.0025 & 0.0339 & 0.0304 & 0.932 & 0.0001 & 0.0226 & 0.0218 & 0.938 \\
  \bottomrule
\end{tabular}
\end{table}

\section{Application}\label{sec:application}

We analyze data collected from the Honolulu Heart Program (HHP, 1965-1990) and the subsequent Honolulu Asia Aging Study (HAAS, 1991-2012), which followed a cohort of men born between 1990 and 1919 and with Japanese ancestry \citep{p2012honolulu, zhang2024marginal, rava2023doubly}. 
We are interested in estimating the cognitive impairment-free survival on the age scale, also known as disease-free survival (DFS) in the time-to-event literature. The data contains subjects who were alive and did not have cognitive impairment when they entered HAAS. Since subjects who ha{d} already died or developed cognitive impairment before HAAS study entry {we}re not included, age to moderate cognitive impairment or death were left truncated by their age at HAAS study entry.  

We applied {the} conditional Kendall's tau test  \citep{tsai1990testing, martin2005testing} to detect potential violation of quasi-independence between age at HAAS study entry and age at moderate cognitive impairment or death; the p-value from the test was 0.0036, providing strong evidence against the quasi-independence assumption between the two ages. 
% \lily{double check if some of the following repeat the intro; if so at least say `as mentioned before' something like that}
As mentioned before, latent factors such as socioeconomic status, health-seeking behavior, and overall health status are likely associated with both the time-to-event and the age at HAAS entry, therefore inducing dependence between the two.
% Note that age at HAAS study entry was determined by the birth cohort and the calendar time of enrollment into HAAS. Risk factors for cognitive impairment and death, such as socioeconomic status, health-seeking behavior, overall health status, and cardiovascular health, may also be associated with the birth cohort, potentially explaining the dependence between the two ages.

We consider the covariates measured during HHP: education ($\leq 12$ years or otherwise), {\it APOE} genotype (present or absence of an {\it APOE E4} risk allele), systolic blood pressure (mmHg), heart rate (beats per 30 seconds),  grip strength (kilograms, at HHP baseline), two measurements of midlife alcohol consumption (heavy/non-heavy, at Exam 1 and Exam 3 during HHP), 
% \lily{2 alc variables?}\yuyao{Yes} 
midlife cigarette consumption (yes/no, at Exam 3 during HHP). 
% \lily{move to appendix: }
After removing subjects with missing covariates, the data contains 1930 subjects who were alive and  not cognitively impaired when they entered HAAS.
Among the 1930 subjects, 351 (18.2\%) were right censored due to loss to follow-up or end of  study. 
The covariate distributions in this data set is summarized in Table \ref{tab:table1} in Appendix \ref{app:HAAS}.
% As \yuyaox{discussed in the introduction and} illustrated in Figure \ref{fig:DAG_HAAS}, 
% grip strength can be viewed as a proxy for the overall health status, which may be associated with both age at study entry and age at moderate cognitive impairment or death. However, the grip strength measurement itself does not directly affect cognitive impairment or death. 
% On the other hand, education, alcohol consumption, and cigarette consumption can be viewed as proxies for socioeconomic status and health-seeking behavior, which may be associated with both age at study entry and age at moderate cognitive impairment or death. While education, alcohol consumption, and cigarette consumption may have a direct effect on cognitive impairment or death, they may only affect age at HAAS study entry through the latent socioeconomic status and health-seeking behavior. 
% Therefore, 
Based on the discussion in Section \ref{sec:preliminary} (with illustration in Figure \ref{fig:DAG_HAAS}),
we classify the covariates into three groups: $W_1$ includes grip strength;
$W_2$ includes education, alcohol consumption, and cigarette consumption;
and $Z$ includes {\it APOE} genotype, systolic blood pressure, and heart rate.

The minimum age at HAAS study entry was 71.3 years old. We therefore focus on estimating the conditional DFS conditional upon surviving to 71.3 years old as the left tail {of} the event time {distribution} is not identifiable before this age \citep{tsai1987note, wang1989semiparametric, wang1991nonparametric}.
%All of the covariates above are measured before 71.3 years old, so they can be viewed as baseline. 
As illustration we estimate the conditional DFS probabilities at ages 80, 85, 90, and 95, that is, $\nu(t) = \ind(t>t_0)$ for $t_0 = 80, 85, 90, 95$, respectively.

We consider the PQB, IPQW, PL, KM, and naive estimators described in Section \ref{sec:simu}. 
Figure \ref{fig:HAAS_survplot} shows the different estimates and their 95\% confidence intervals based on {the} bootstrapped standard errors. 
% We see that the PQB estimates are lower than the other estimates, indicating that there are probably unmeasured latent factors that induce the dependence between $Q^*$ and $T^*$. 

% Our analysis reveals differences in estimated DFS for the cohort of men with Japanese ancestry compared with existing approaches in the literature that do not account for unmeasured dependence between age at study entry and DFS. 
% Specifically, the PQB estimates are substantially lower than those obtained from other methods, especially for survival probabilities at earlier ages (aound 80 - 90). 
We observe notable differences in estimated DFS for men with Japanese ancestry relative to methods that ignore unmeasured dependence between age at study entry and time-to-event, with PQB yielding substantially lower estimates, especially at earlier ages (around 80–90).
% a more nuanced survival profile for the HAAS cohort than previously estimated. Specifically, the PQB estimates yield a significantly lower DFS curve compared to standard IPQW and PL methods, particularly during the early follow-up stages (ages 80-85).
These findings indicate the presence of unmeasured latent factors inducing the dependence between age at study entry and DFS, and suggest that commonly used methods may underestimate the risk of cognitive decline and mortality at earlier ages due to inadequate adjustment for this dependence. 
By leveraging grip strength, education, and midlife lifestyle variables as proxies for 
% dependence-inducing latent factors 
the latent socioeconomic status, health-seeking behavior, and overall health status, the PQB estimator helps account for such unmeasured dependence. The resulting lower DFS estimates from the PQB estimator highlight the risk of cognitive decline and mortality at earlier ages that may not be captured by the commonly used methods such as IPQW and PL. This has potential implications for the timing of preventive interventions for cognitive decline and mortality risk, as well as the planning and allocation of public health resources related to brain aging and mortality.

% From a clinical perspective, these findings imply that the cognitive decline process in this specific population may begin earlier or progress more rapidly than indicated by traditional survival models. 
% The standard methods likely over-adjust for truncation by assuming quasi-independence, effectively ``filtering out'' individuals who were at higher latent risk but entered the study at younger ages. 
% By leveraging midlife grip strength as a physical frailty proxy, the PQB estimator recovers this latent risk, suggesting that preventive screening for Alzheimer’s disease and cognitive aging should perhaps be initiated earlier than currently suggested by biased population estimates. Furthermore, the gap between PQB and PL estimates reinforces the hypothesis that socioeconomic and health-seeking behavior—captured here via education and lifestyle proxies—are not merely ``adjustments'' but are intrinsically tied to the study entry mechanism itself.
% Ignoring this unmeasured dependence leads to an over-optimistic assessment of brain aging, which could misguide public health resource allocation for aging-in-place services in the Japanese-American community.

We also observe that the PQB estimator has wider confidence intervals compared to the IPQW estimator, coincides with the observation in simulation.
The PL estimates are higher than the PQB and the IPQW estimates, especially at later ages, likely due to violation of the quasi-independence assumption for left truncation. 
The KM and naive estimates, which ignore left truncation, are substantially higher than the PQB, IPQW, and PL estimates, reflecting the selection bias due to left truncation. 
Consequently, conclusions drawn from these approaches will substantially overestimate the cognitive impairment-free survival in the HAAS cohort.
The naive estimates, which ignore both left truncation and right censoring, are lower than the KM estimates, which ignore only left truncation. This reflects the additional bias from ignoring right censoring, which is in the opposite direction to the bias induced by left truncation. 

\begin{figure}[ht]
\centering
 \includegraphics[width=0.75\textwidth]{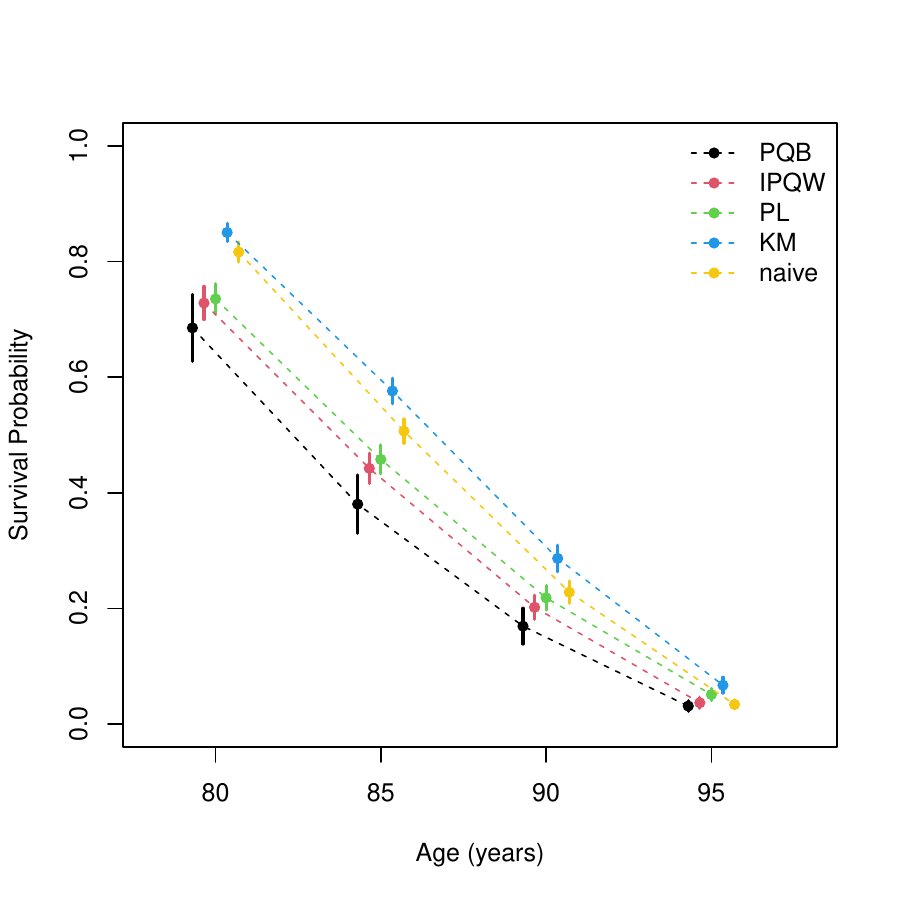}
\caption{Estimated DFS at ages 80, 85, 90, 95 from different estimators.}
\label{fig:HAAS_survplot}
\end{figure}

\section{Discussion}\label{sec:discussion}

% \andrew{Move this to future directions} Both Assumptions \ref{ass:proximal_indep} and \ref{ass:positivity} are untestable since they involve the unmeasured $U^*$.

% \andrew{Only thing we want to sell is truncation, but spent too much time on censoring. Should discuss future direction for truncation. }

We have proposed a proximal weighting  identification and estimation framework for handling dependent left truncation, which admits measured covariates may only serve as imperfect proxies for explaining the dependence between the left truncation time and the event time. 
To our best knowledge, this is the first work that leverages proxies to handle dependent left truncation.
% \lily{is this true?}\yuyao{No, there has been work for dependent left trucation for example under copula models or assuming a latent random left truncation time after a transformation. Leveraging proxies for nonparametric identification under dependent left truncation is the first.} 
The importance of the proposed proximal framework is underscored by our findings in the HAAS data. Existing methods, such as IPQW and PL, yield survival estimates that appear overly optimistic, likely because they fail to account for the unmeasured dependence between age at study entry and cognitive impairment-free survival. By leveraging grip strength, education, and midlife lifestyle variables as proxies, our approach provides a more plausible assessment of cognitive impairment-free survival, suggesting that the 
% ``healthy entrant" 
selection
bias due to delayed entry in aging studies may be driven by latent factors that existing approaches cannot eliminate.

The proposed identification relies on the key assumption of proximal independence, which in practice needs to be justified by domain knowledge and is untestable due to the involvement of unmeasured latent factors.  
Developing sensitivity analysis for potential violations of this assumption is an interesting direction for future research.

The proposed estimator requires estimating the truncation-inducing bridge process involved in the identification. We have illustrated the estimation of the bridge process using a semiparametric working model with an additive form, under which a closed-form solution exists. 
As mentioned earlier, there have been recent developments in nonparametric estimation of bridge functions in proximal causal inference settings, including adversarial learning \citep{ghassami2022minimax, kallus2021causal, olivas2025source} and debiased ill-posed regression \citep{ghassami2025debiased}. Extending these approaches to estimate the bridge processes for time-to-event settings is another promising future direction.

% \lily{keep this}
 As discussed in \citet{qian2014assumptions}, there are two censoring scenarios in the presence of left truncation: 1) censoring may happen before left truncation, and 2) censoring always after left truncation. We developed an approach for the second censoring scenario because it is the case for our application. Similar approaches can be developed for the first censoring scenario, by considering censoring on the original time scale.
 Our handling of right censoring  assumes independent
% % \lily{double check earlier, `random' for censoring don't mean the same thing as for truncation}\yuyao{Checked}
 residual censoring. Under the weaker conditional {independent} censoring assumption given the observed covariates, IPCW can be similarly applied conditional on the covariates. In addition, if there is dependence between the  censoring time and the event time of interest after conditioning on the measured covariates, the proximal identification framework proposed in \citet{ying2024proximal} can be incorporated, again by considering censoring on the residual time scale. We leave this direction for future work.

The code for implementation of the proposed estimators and simulations in this paper is available at \url{https://github.com/wangyuyao98/truncProxy_weighting}.

% \section{Notes: Additional things to think about on the side}

% \td{May consider extending to censoring before truncation scenario as well and add the CNS lymphoma data example? How much does it add to the paper? Whether it is worth adding it?} \yuyao{If proceed, will need to check via simulation if the sample size 98 is too small.}

% \qq{Can we come up with a better approach for solving a bridge equation for the truncation and censoring weights together?} -- \yuyaox{I'm thinking maybe not for this censoring case. Because the censoring weight needed is conditional on $Q$ (note that $S_D(T-Q) = \PP(C>T|Q,T)$), while the truncation weights is related to the $Q$ distribution and is not conditional on $Q$. So I think even in the degenerating case with $(Z,U)$ both observed, we may not be able to compute the weights by solving a single bridge equation like \eqref{eq:trunc_bridge_UZ}.}

\section*{Acknowledgments}

This work has used the computing services provided by the Open Science Grid (OSG) Consortium, 
% \citep{osg07, osg09, https://doi.org/10.21231/906p-4d78, https://doi.org/10.21231/0kvz-ve57}, 
which is supported by the National Science Foundation awards \#2030508 and \#1836650.
AI tools (chatGPT) were used to assist in converting part of the code for implementation from R to C++, which has significantly reduced the computation time of the proposed estimator. 
AI tools were also used to assist with editing, including correcting grammatical errors and improving the clarity and flow of the text.

% Bibliography---------------------------------------------------------------------------------------------------------------------------------------------%
% \clearpage
%\spacingset{1.0}
% \baselineskip 16pt
% \parskip 0pt

\bibliographystyle{biom}
\bibliography{bib/left_trunc, bib/proxy, bib/cox, bib/osg}

\appendix

\newpage
\section{Preliminaries}

Here we introduce the {\it reverse time hazard function} of $Q^*$ and its associated properties, which will be used in the later proofs. 

Recall that $G(t|z,u)$ denotes the conditional cumulative distribution function (CDF) of $Q^*$ given $(Z^*,U^*)$.
Let $\bar\lambda_Q$ denote the conditional {\it reverse time hazard function} of $Q^*$ given $(Z^*,U^*)$: 
\begin{align}
\bar\lambda_Q(q|z,u) &= \lim_{h\to 0+} \frac{\PP(q-h< Q^*\leq q|Q^*\leq q, Z^* = z, U^* = u)}{h}= \frac{dG(q|z,u)}{G(q|z,u)}. \label{eq:lambda_bar} 
\end{align}
 Let $p(q|z,u)$ denote the conditional density function of $Q$ given $(Z,U)$ in the observed data.
Following the proof in \citet{wang2024doubly} (with $(Z^*,U^*)$ being the covariates in their paper), we have 
\begin{align}
    G(q|z,u) & = \exp\left\{ - \int_q^\tauQ \bar\lambda_Q(t|z,u) dt \right\}, \label{eq:G_lambda} \\
    \bar\lambda_Q(q|z,u) & = \frac{p(q|z,u)}{\PP(Q\leq q<T\mid Z=z,U=u)}. \label{eq:lambda_p}
\end{align}
Let $\prodi$ denote the integral product. Then \eqref{eq:G_lambda} can be expressed as
\begin{align}
    G(q|z,u) = \Prodi_{q}^{\tauQ} \{1-\bar\lambda_Q(t|z,u) dt\}.  \label{eq:G_lambda_prodi} 
\end{align}
In addition, \eqref{eq:lambda_p} implies
\begin{align}
    \PP(q-dq < Q \leq q \mid Q\leq q<T, Z=z,U=u) 
    & = \frac{p(q|z,u)dq}{\PP(Q\leq q<T \mid  Z=z,U=u)} \\
    & = \bar\lambda_Q(q|z,u) dq. \label{eq:lambda_d_p}
\end{align}

In the following, we will use ``L.H.S" as a short hand for ``left hand side" and ``R.H.S" as a short hand for ``right hand side". 
Let $\bar N_Q^*(t) = \ind(t\leq Q^*)$.

\newpage
\section{Additional discussion and plausibility of Assumption \ref{ass:trunc_bridge}}\label{app:plausibility}

Equation \eqref{eq:trunc_bridge_def} in Assumption \ref{ass:trunc_bridge} resembles the equation that defines the censoring-inducing bridge process in \citet{ying2024proximal} for handling dependent right censoring,
which is expressed in terms of increments of the censoring-inducing bridge process and the censoring counting process, with the expectation taken conditional on being ``at risk" for each time $t$, which is $\{T\geq t, C\geq t\}$ in right censored data. 
In comparison, our equation \eqref{eq:trunc_bridge_def} involves the backwards counting process for $Q$, and the expectation is conditional on $Q\leq t<T$, the event of being ``at-risk" at time $t$ for left truncated data. 

As mentioned in the main paper, we take the existence of the truncation-inducing bridge process as a primitive assumption.
% In the proximal causal inference literature \citep{miao2018identifying, ying2023proximal, cui2024semiparametric, tchetgen2024introduction}, similar equations are used to define bridge functions. However, those bridge functions are functions of finite-dimensional random variables, instead of stochastic processes that involve the time dimension, as in this work and in \citet{ying2024proximal}. 
% In proximal causal inference, the bridge equations are Fredholm integral equations, and the existence of their solutions is given by Picard's theorem under certain regularity conditions \citep{miao2018identifying, ying2023proximal, cui2024semiparametric}. 
% Investigating the existence of such a process for left truncated data would be of interest in probability theory but is beyond the scope of the current paper. 
Nevertheless, to assure the reader, we provide below a data generating mechanism under which Assumption \ref{ass:trunc_bridge} holds. 

Suppose $(Z^*,U^*)$ follow a multivariate normal distribution, and suppose that $(W_1^*,W_2^*)$ are generated from: 
\begin{align}
    W_1^* & = \gamma_{11} U^* + \gamma_{12} Z^* + \epsilon_1 , 
    \qquad W_2^* = \gamma_{21} U^* + \gamma_{22} Z^* + \epsilon_2,
\end{align}
where $\epsilon_1\sim N(0,\sigma_1^2)$, $\epsilon_2\sim N(0,\sigma_2^2)$ are independent random variables that are also independent of $(Z^*,U^*)$. 
Let $\lambda_T$ denotes the conditional hazard function of $T^*$ given $(Z^*,U^*)$, and recall $\bar\lambda_Q$ from \eqref{eq:lambda_bar} the conditional reverse time hazard function of $Q^*$ given $(Z^*,U^*)$.
$T^*$ and $Q^*$ are generated from the following additive hazard model: 
\begin{align}
    \lambda_{T}(t|Z^*,U^*) & = \beta_0(t) + \beta_z(t) Z^* + \beta_u(t) U^*, \\ 
    \bar\lambda_{Q}(t|Z^*,U^*) & = \alpha_0(t) + \alpha_z(t) Z^* + \alpha_u(t) U^*. \label{eq:DGM_lambda_Q}
\end{align}
We show that under this data generating mechanism, there exists a stochastic process $b$ satisfying the conditions \eqref{eq:trunc_bridge_def} and \eqref{eq:trunc_bridge_initial_condi}.

Consider
\begin{align}
    b(t,W_1,Z;B(t)) = \exp\{B_0(t) + W_1 B_1(t) + Z B_z(t) \}, \label{eq:model_b_supp}
\end{align}
where $B_0(t)$, $B_1(t)$, and $B_z(t)$ are arbitrary bounded functions of $t \in[0,\tauQ]$ with the initial condition $B_0(\tauQ) =  B_1(\tauQ) = B_z(\tauQ) = 0$. Under this initial condition, \eqref{eq:trunc_bridge_initial_condi} holds. We now show that there exists $B_0(t)$, $B_1(t)$, $B_z(t)$ such that \eqref{eq:trunc_bridge_def} holds.

By the tower property of expectations, the left hand side (L.H.S.) of \eqref{eq:trunc_bridge_def}: 
\begin{align*}
    &\quad \E\{\dd b(t, W_1,Z) - \dd \bar N_Q(t) b(t,W_1,Z) \mid Q\leq t<T, W_2,Z\} \\
    & = \quad \E\left[\E\left\{\dd b(t, W_1,Z) - \dd \bar N_Q(t) b(t,W_1,Z) \mid  Q\leq t<T, Z,U, W_1, W_2 \right\} \mid Q\leq t<T, W_2, Z\right] \\
    & = \quad \E\left[ \dd b(t, W_1,Z) - \E\left\{\dd \bar N_Q(t)  \mid  Q\leq t<T, Z,U, W_1, W_2 \right\} b(t,W_1,Z) \mid Q\leq t<T, W_2, Z\right]. 
\end{align*}
Recall $\bar N_Q^*(t) = \ind(Q^*\geq t)$. 
Since $(Q,T,W_1,W_2,Z,U)$ has the same distribution as \\
$(Q^*,T^*,W_1^*,W_2^*,Z^*,U^*)\mid Q^*<T^*$, we have 
\begin{align}
    \E\left\{\dd \bar N_Q(t) \mid Q\leq t<T, Z,U, W_1, W_2 \right\}
    & = \E\left\{\dd \bar N_Q^*(t) \mid Q^*\leq t<T^*, Z^*,U^*, W_1^*, W_2^* \right\} \nonumber \\
    & = \E\left\{\dd \bar N_Q^*(t) \mid Q^*\leq t<T^*, Z^*,U^*\right\} \label{eq:compat_proof_1} \\
    & = \E\left\{\dd \bar N_Q(t) \mid Q\leq t<T, Z,U\right\} \nonumber\\
    & = \PP(t-dt<Q\leq t\mid Q\leq t<T,Z,U) \label{eq:compat_proof_5}  \\
    & = -\bar\lambda_Q(t|Z,U)dt, \nonumber
\end{align}
where \eqref{eq:compat_proof_1} holds because $Q^*\bigCI (W_1^*,W_2^*) \mid Z^*,U^*$ under this data generating mechanism, and \eqref{eq:compat_proof_5} holds by \eqref{eq:lambda_d_p}.
Therefore, $b$ satisfies \eqref{eq:trunc_bridge_def} in Assumption \ref{ass:trunc_bridge} if and only if 
\begin{align}
    \E\left[ \dd b(t,W_1,Z) + \bar\lambda_Q(t|Z,U) b(t,W_1,Z)  \mid Q\leq t<T, W_2, Z\right] = 0. \label{eq:compat_proof_2}
\end{align}
By plugging \eqref{eq:DGM_lambda_Q} and \eqref{eq:model_b_supp} into \eqref{eq:compat_proof_2}, we have 
\begin{align}
    &\E\big[ \{dB_0(t) + dB_1(t) W_1 + dB_z(t) Z\} \exp\{B_0(t) + B_1(t) W_1 + B_z(t) Z\}  \\
    &\quad +  \{\alpha_0(t) + \alpha_z(t) Z + \alpha_u(t) U\} \exp\{B_0(t) + B_1(t) W_1 + B_z(t) Z\}  \mid Q\leq t<T, W_2, Z \big] =0.
    \label{eq:compat_proof_3}
\end{align}
Denote 
$$\Ec_\ell(t,w_2,z,u) = \E\left[ W_1^\ell\exp\{B_1(t) W_1\} \mid Q\leq t<T, W_2 = w_2, Z = z, U = u\right],$$ 
for $\ell = 0,1$. 
Then by tower property of expectations and some algebra,  
\begin{align}
    \text{L.H.S. of \eqref{eq:compat_proof_3}} 
    & = \A \cdot \exp\{B_0(t) + B_z(t) Z\}, \label{eq:compat_proof_4}
\end{align}
where 
\begin{align*}
    \A 
    & = \{dB_0(t) + dB_z(t) Z + \alpha_0(t) + \alpha_z(t) Z + \alpha_u(t) U\} \cdot \Ec_0(t,W_2,Z,U) \\
    &\quad + dB_1(t) \cdot \Ec_1(t,W_2,Z,U).
\end{align*}
We now compute $\Ec_\ell(t,W_2,Z,U)$ for $\ell = 0,1$.
Since the observed data $(Q,T,W_1,W_2,Z,U)$ has the same distribution as $(Q^*,T^*,W_1^*,W_2^*,Z^*,U^*)\mid Q^*<T^*$, we have
\begin{align*}
     % \E\left\{\left. W_1^\ell \ e^{B_1(t) W_1} \right| Q\leq t<T, Z,U,W_2 \right\} 
     \Ec_\ell(t,w_2,z,u)
     & = \E\left\{\left. (W_1^*)^\ell \ e^{B_1(t) W_1^*} \right| Q^*\leq t<T^*, W_2^* = w_2, Z^* = z, U^* = u \right\} \\
     & = \E\left\{\left. (W_1^*)^\ell  \ e^{B_1(t) W_1^*} \right| Z^* = z, U^* = u \right\},
\end{align*}
where the last equality holds because $W_1^*\bigCI (W_2^*,Q^*,T^*)\mid (Z^*,U^*)$ under this data generating mechanism. 
Since the conditional distribution $W_1^* \mid (Z^*,U^*)$ is normal with mean $\gamma_{11} U^* + \gamma_{12} Z^*$  and variance $\sigma_1^2$, by properties of the moment generating function for normal distributions, we have 
\begin{align}
    \Ec_0(t,w_2,z,u)
    & = \E\left\{\left. e^{B_1(t) W_1^*} \right| Z^* = z, U^* = u \right\}\\
    & = \exp\left\{ (\gamma_{11} u + \gamma_{12} z)B_1(t) + \frac{1}{2}\sigma_1^2 B_1(t)^2\right\}, \\
    \Ec_1(t,w_2,z,u)
    & = \E\left\{\left. W_1^*  \ e^{B_1(t) W_1^*} \right| Z^* = z, U^* = u \right\} \\
    & = \{\gamma_{11}u + \gamma_{12}z + \sigma_1^2 B_1(t)\} \exp\left\{ (\gamma_{11}u + \gamma_{12}z)B_1(t) + \frac{1}{2}\sigma_1^2 B_1(t)^2\right\}. 
\end{align}
Plugging the above into \eqref{eq:compat_proof_3} \eqref{eq:compat_proof_4} and by comparing the coefficients for the terms involving $U^*$ and $Z^*$ and the constant terms, 
it suffices to find $B_0(t),B_z(t),B_1(t)$ satisfying 
\begin{align}
    \left\{ 
    \begin{array}{l}
        \gamma_{11} dB_1(t) + \alpha_u(t) = 0, \\
        dB_z(t) + \gamma_{12} dB_1(t) + \alpha_z(t) = 0, \\
        dB_0(t) + \alpha_0(t) + \sigma_1^2 B_1(t) dB_1(t) = 0, \\
    \end{array}
    \right.
\end{align}
in order for \eqref{eq:trunc_bridge_equation} to hold. 
Solving the above equations gives
\begin{align}
    \left\{ 
    \begin{array}{l}
        dB_1(t) = - \gamma_{11}^{-1} \alpha_u(t), \\
        dB_z(t) = - \alpha_z(t) - \gamma_{12} dB_1(t), \\
        dB_0(t) = - \alpha_0(t) - \sigma_1^2 B_1(t) dB_1(t).
    \end{array}
    \right. \label{eq:model_compat_5}
\end{align}
Therefore, if we take 
\begin{align*}
    \left\{ 
    \begin{array}{l}
        B_1(t) = \gamma_{11}^{-1} \int_t^{\tauQ} \alpha_u(s) ds, \\
        B_z(t) = \int_t^{\tauQ}  \alpha_z(s) ds + \gamma_{12}\int_t^\tau dB_1(s), \\
        B_0(t) = \int_t^{\tauQ} \alpha_0(s)ds + \sigma_1^2 \int_t^{\tauQ}  B_1(s) dB_1(s),
    \end{array}
    \right. 
\end{align*}
then the bridge process in \eqref{eq:model_b_supp} satisfies the conditions in \eqref{eq:trunc_bridge_def} and \eqref{eq:trunc_bridge_initial_condi}.

\newpage
\section{Identification proofs}

\subsection{Proof of Lemma \ref{thm:proximal_trunc_identification}}\label{app:proof_identification}

\begin{proof}[Proof of Lemma \ref{thm:proximal_trunc_identification}]
    1) We first show \eqref{eq:trunc_bridge_UZ}. 
    % We will use ``L.H.S." as a shorthand for ``left hand side". 
    Let $\xi(t,Z,U)$ denote the left hand side of \eqref{eq:trunc_bridge_UZ}, i.e., 
    $$\xi(t,Z,U) = \E\{\dd b(t, W_1,Z) - \dd \bar N_Q(t) b(t,W_1,Z) \mid Q\leq t<T, Z, U\}.$$ 
    By Assumption \ref{ass:completeness_truncbridge}, it suffices to show that 
    \begin{align*}
        \E\left[\xi(t,Z,U) \mid Q\leq t<T, W_2,Z\right] = 0.
    \end{align*}
    Recall that $\bar N_Q^*(t) = \ind(t\leq Q^*)$. Since $(Q,T,W_1,W_2,Z,U)$ has the same distribution as \\
    $(Q^*,T^*,W_1^*,W_2^*,Z^*,U^*)\mid Q^*<T^*$, we have 
    \begin{align}
        &\quad \E\left[\xi(t,Z,U) \mid Q\leq t<T, W_2,Z\right] \\
        & = \E\left[ \E\{\dd b(t, W_1,Z) - \dd \bar N_Q(t) b(t,W_1,Z) \mid Q\leq t<T,Z,U\}  \mid Q\leq t<T, W_2,Z \right] \\
        & = \E\left[ \E\{\dd b(t, W_1^*,Z^*) - \dd \bar N_Q^*(t) b(t,W_1^*,Z^*) \mid Q^*\leq t<T^*,Z^*,U^*\}  \mid Q^*\leq t<T^*, W_2^*,Z^* \right] \\
        & = \E\left[ \E\{\dd b(t, W_1^*,Z^*) - \dd \bar N_Q^*(t) b(t,W_1^*,Z^*) \mid Q^*\leq t<T^*, W_2^*,Z^*,U^*\}  \mid Q^*\leq t<T^*, W_2^*,Z^* \right] \label{eq:truncProxy_proof3}\\
        & = \E\left[ \E\{\dd b(t, W_1,Z) - \dd \bar N_Q(t) b(t,W_1,Z) \mid Q\leq t<T, W_2,Z,U\}  \mid Q\leq t<T, W_2,Z \right] \\
        & = \E\{\dd b(t, W_1,Z) - \dd \bar N_Q(t) b(t,W_1,Z) \mid Q\leq t<T, W_2, Z\} \\
        & = 0, \label{eq:lem1_proof_1}
    \end{align}
    % \blue{
    % \begin{align}
    %     &\quad \E\left[ \E\{\dd b(t, W_1,Z) - \dd \bar N_Q(t) b(t,W_1,Z) \mid Q\leq t<X, Z,U\}  \mid Q\leq t<X, W_2,Z \right] \\
    %     & = \E\left[ \E\{\dd b(t, W_1^*,Z^*) - \dd \bar N_Q^*(t) b(t,W_1^*,Z^*) \mid Q^*\leq t<X^*, Z^*,U^*\}  \mid Q^*\leq t<X^*, W_2^*,Z^* \right] \\
    %     & = \E\left[ \E\{\dd b(t, W_1^*,Z^*) - \dd \bar N_Q^*(t) b(t,W_1^*,Z^*) \mid Q^*\leq t<X^*, W_2^*,Z^*,U^*\}  \mid Q^*\leq t<X^*, W_2^*,Z^* \right] \label{eq:truncProxy_proof3}\\
    %     & = \E\left[ \E\{\dd b(t, W_1,Z) - \dd \bar N_Q(t) b(t,W_1,Z) \mid Q\leq t<X, W_2,Z,U\}  \mid Q\leq t<X, W_2,Z \right] \\
    %     & = \E\{\dd b(t, W_1,Z) - \dd \bar N_Q(t) b(t,W_1,Z) \mid Q\leq t<X, W_2, Z\} \\
    %     & = 0,
    % \end{align}
    % }
    where \eqref{eq:truncProxy_proof3} holds by Assumption \ref{ass:proximal_indep} and \eqref{eq:lem1_proof_1} holds by \eqref{eq:trunc_bridge_def}.

    2) We now show \eqref{eq:identification}. 
    By the towel property of expectations, we have 
    \begin{align}
        \E\{b(T,W_1,Z) \nu(T)\} 
        & = \E[ \E\{b(T,W_1,Z)\mid Q,T,Z,U\} \nu(T)]. \label{eq:truncProxy_proof5}
    \end{align}
    In the following, we compute the inner conditional expectation in the R.H.S. of \eqref{eq:truncProxy_proof5}. 
    Recall that $\prodi$ denotes the integral product. 
    Since in the observed data $Q<T$ %\blue{$Q<X$} 
    a.s., by \eqref{eq:trunc_bridge_initial_condi} and the property of integral products, we have 
    \begin{align}
        \E\{b(T,W_1,Z)\mid Q,T,Z,U\}
        & = \E\{b(T,W_1,Z)\mid Q<T,Q,T,Z,U\} \\
        & = \Prodi_{T}^{\tauQ} \frac{\E\{b(s,W_1,Z)\mid Q\leq s-ds,T,Z,U\}}{\E\{b(s+ds,W_1,Z)\mid Q\leq s,T,Z,U\}}. \label{eq:truncProxy_proof1}
    \end{align}
    We now compute the conditional expectations in the R.H.S. of \eqref{eq:truncProxy_proof1}. 
    Since $(Q,T,W_1,W_2,Z,U)$ has the same distribution as $(Q^*,T^*,W_1^*,W_2^*,Z^*,U^*)\mid Q^*<T^*$, we have 
    \begin{align}
        &\quad \E\{b(s,W_1,Z)\mid Q\leq s-ds,T,Z,U\} \\
        & = \E\{b(s,W_1^*,Z^*)\mid Q^*<T^*, Q^*\leq s-ds, T^*,Z^*,U^*\}  \\
        % & = \E\{b(s,W_1^*,Z^*)\mid Q^*<T^*, Q^*\leq s-ds,T^*>s,Z^*,U^*\} \\ 
        & = \E\{b(s,W_1^*,Z^*)\mid Q^*<T^*, Q^*\leq s-ds, T^*>s,Z^*,U^*\} \label{eq:truncProxy_proof4} \\
        & = \E\{b(s,W_1,Z)\mid Q\leq s-ds,T>s,Z,U\}.   
    \end{align}
    where \eqref{eq:truncProxy_proof4} holds because $(W_1^*, Q^*) \bigCI T^* \mid (Z^*,U^*)$, which is implied from Assumption \ref{ass:proximal_indep}. 
    Likewise, 
    \begin{align}
        \E\{b(s+ds,W_1,Z)\mid Q\leq s,T,Z,U\} = \E\{b(s+ds,W_1,Z)\mid Q\leq s,T>s,Z,U\}.
    \end{align}
    These together with \eqref{eq:truncProxy_proof1} imply
    \begin{align}
        \E\{b(T,W_1,Z)\mid Q,T,Z,U\}
        & = \Prodi_{T}^{\tauQ} \frac{\E\{b(s,W_1,Z)\mid Q\leq s-ds, T>s, Z,U\}}{\E\{b(s+ds,W_1,Z)\mid Q\leq s, T>s, Z,U\}}. \label{eq:truncProxy_proof2}
    \end{align}
    We now compute the fraction in the R.H.S. of \eqref{eq:truncProxy_proof2} using \eqref{eq:trunc_bridge_UZ}. 
    Recall that $\bar N_Q(t) = \ind(t\leq Q<T)$, so $\bar N_Q(t+dt) = 0$ for $t\geq Q$. 
    Note that $db(t,W_1,Z) = b(t+dt,W_1,Z) - b(t,W_1,Z)$ and $d\bar N_Q(t) = \bar N_Q(t+dt) - \bar N_Q(t)$. 
    Therefore, 
    \begin{align}
        &\quad \text{L.H.S. of \eqref{eq:trunc_bridge_UZ}} \\
        & = \E\big[\{b(t+dt,W_1,Z) - b(t,W_1,Z)\} - \{\bar N_Q(t+dt) - \bar N_Q(t)\} b(t,W_1,Z;\theta) \mid Q\leq t<T, Z,U \big] \\
        & = \E\left[b(t+dt,W_1,Z) - \{1-\bar N_Q(t)\} b(t,W_1,Z) \mid Q\leq t<T, Z,U \right] \\
        % & =  \E\{b(t+dt, W_1,Z )\mid Q\leq t<T, Z,U\} -  \E[\{1-\bar N_Q(t)\} b(t, W_1,Z)\mid Q\leq t<T, Z,U]   \\
        & = \E\{b(t+dt, W_1,Z)\mid Q\leq t<T, Z,U\} 
        - \E\{ \ind(Q \leq t-dt) b(t, W_1,Z)\mid Q\leq t<T, Z,U\} \\
        & = \E\{b(t+dt, W_1,Z)\mid Q\leq t, T>t, Z,U\} \\
        &\quad - \E\{b(t, W_1,Z)\mid Q\leq t-dt, T>t, Z,U\} \cdot \PP(Q\leq t-dt \mid Q\leq t<T, Z,U). 
    \end{align}
    So \eqref{eq:trunc_bridge_UZ} implies
    \begin{align}
        &\quad \E\{b(t+dt, W_1,Z)\mid Q\leq t, T>t, Z,U\} \\
        & =  \E\{b(t, W_1,Z)\mid Q\leq t-dt, T>t, Z,U\} \cdot \PP(Q\leq t-dt \mid Q\leq t<T, Z,U). \label{eq:lem1_proof_2}
    \end{align}
    % Recall $\bar\lambda_Q$ from \eqref{eq:lambda_bar} the conditional reverse time hazard function of $Q^*$ given $(Z^*,U^*)$ and its property in \eqref{eq:lambda_d_p}. 
    % Equation \eqref{eq:lem1_proof_2} implies
    Therefore, 
    \begin{align}
        \frac{\E\{b(t, W_1,Z)\mid Q\leq t-dt, T>t, Z,U\}}{\E\{b(t+dt, W_1,Z)\mid Q\leq t, T>t, Z,U\}} & = \frac{1}{\PP(Q\leq t-dt \mid Q\leq t<T, Z,U)}   \\
        & = \frac{1}{1-\PP(Q> t-dt \mid Q\leq t<T, Z,U)} \\ 
        & = \frac{1}{1-\PP(t-dt < Q \leq t \mid Q\leq t<T, Z,U)} \\
        & = \frac{1}{1-\bar\lambda_Q(t|Z,U) dt}, \label{eq:lem1_proof_3}
    \end{align}
    where \eqref{eq:lem1_proof_3} holds by \eqref{eq:lambda_d_p}. 
    This combined with \eqref{eq:truncProxy_proof1} implies
    \begin{align}
        \E\{b(T,W_1,Z) \nu(T)\} 
        & = \E\left[ \frac{\nu(T)}{\prodi_{T}^{\tauQ} \{1-\bar\lambda_Q(t|Z,U) dt\}} \right] 
        = \E\left\{\frac{\nu(T)}{G(T|Z,U)}\right\}, \label{eq:lem1_proof_4}
    \end{align}
    where the last equation holds by \eqref{eq:G_lambda_prodi}.  
    Following the proof of Lemma 1 in \citet[Supplementary Material]{wang2024doubly}, we have 
    \begin{align*}
        \E\left\{\frac{\nu(T)}{G(T|Z,U)}\right\} = \beta^{-1} \E\{\nu(T^*)\}, 
    \end{align*}
    where $\beta = \PP(Q^*<T^*)$.
    This together with \eqref{eq:lem1_proof_4} imply
    \begin{align*}
        \E\{b(T,W_1,Z) \nu(T)\} = \beta^{-1} \E\{\nu(T^*)\}. 
    \end{align*}
    As a special case, when $\nu(t) = 1$ for all $t$, we have
    \begin{align*}
        \E\{b(T,W_1,Z)\}  
        = \beta^{-1}.
    \end{align*}
    Therefore, 
    \begin{align*}
        \frac{\E\{b(T,W_1,Z) \nu(T)\}}{\E\{b(T,W_1,Z)\}} 
        = \E\{\nu(T^*)\}
         = \theta. 
    \end{align*}
    
\end{proof}

\subsection{Inverse probability of truncation weighting as a special case}\label{app:special_case}

We now prove the claim in Remark \ref{rm:IPQW}. 
Consider the case with $U = \varnothing, W_1 = \varnothing, W_2 = \varnothing$ so that $Q^*\bigCI T^*\mid Z^*$. We will verify that $b(t,Z) = 1/G(t|Z)$ satisfies \eqref{eq:trunc_bridge_def}, where $G(t|z)$ in this case denotes the conditional cumulative distribution function of $Q^*$ given $Z^* = z$. In other words, we will show that 
\begin{align*}
    \E\{\dd b(t,Z) - \dd \bar N_Q(t) b(t,Z) \mid Q\leq t<T, Z\} = 0.
\end{align*}

Recall the definition of $\bar\lambda_Q$ in \eqref{eq:lambda_bar}. In the following, we will use the notation $\bar\lambda_Q$ in the context of $U^* = \varnothing$. 
We have 
\begin{align*}
    db(t,Z) = - \frac{dG(t|Z)}{G(t|Z)^2} = -\frac{\bar\lambda_Q(t|Z) dt}{G(t|Z)}. 
\end{align*}
On the other hand,
\begin{align*}
    \E\left\{\left. \dd \bar N_Q(t) \right| Q\leq t<T, Z \right\}
    % & =  \E\left\{\left. \bar N_Q(t) - \bar N_Q(t-dt) \right| Q\leq t<T, Z \right\} \\
    & = - \E\left\{\left. \ind(t-dt < Q \leq t) \right| Q\leq t<T, Z \right\} \\
    & = - \PP(t-dt < Q \leq t \mid  Q\leq t<T, Z ) \\
    % & = - \PP(t-dt < Q^* \leq t \mid  Q^*\leq t<T^*, Z^* ) \\
    % & = - \PP(t-dt < Q^* \leq t \mid  Q^*\leq t, Z^* ) \\
    & = - \bar\lambda_Q(t|Z) dt,
\end{align*}
where the last equality holds by \eqref{eq:lambda_d_p}.
Combining the above, we have 
\begin{align*}
    \E\{\dd b(t,Z) - \dd \bar N_Q(t) b(t,Z) \mid Q\leq t<T, Z\}  
    & = \dd b(t,Z) - \E\{\dd \bar N_Q(t) \mid Q\leq t<T, Z\} b(t,Z)
    = 0.
\end{align*}

% \newpage
\subsection{Proof of Lemma \ref{thm:identification_cen_after}}\label{app:identification_proof_cen_after}

\begin{proof}[Proof of Lemma \ref{thm:identification_cen_after}]
    
    When $\Delta = 1$, we have $X = T$; and recall that $C = Q + D$.
    So
    \begin{align}
        \E\left\{\frac{\Delta b(X,W_1,Z)\nu(X)}{S_D(X-Q)} \right\}
        & = \E\left\{\frac{\ind(T<C) b(T,W_1,Z)\nu(T)}{S_D(T-Q)} \right\} \nonumber \\
        & = \E\left[ \E\left\{ \left. \frac{\ind(T<C) b(T,W_1,Z)\nu(T)}{S_D(T-Q)} \right| Q,T,W_1,Z \right\} \right] \nonumber \\
        & = \E\left[ \frac{b(T,W_1,Z)\nu(T)}{S_D(T-Q)}  \cdot \E\left\{ \ind(T<C) \left. \right| Q,T,W_1,Z \right\} \right] \nonumber \\
        & = \E\left[ \frac{b(T,W_1,Z)\nu(T)}{S_D(T-Q)}  \cdot \E\left\{ \ind(D > T-Q) \left. \right| Q,T,W_1,Z \right\} \right].
    \end{align}
    By Assumption \ref{ass:cen_random},
    \begin{align*}
        \E\left\{ \ind(D > T-Q) \left. \right| Q,T,W_1,Z \right\} = S_D(T-Q), 
    \end{align*}
    so 
    \begin{align*}
        \E\left\{\frac{\Delta b(X,W_1,Z)\nu(X)}{S_D(X-Q)} \right\}
        = \E\left\{ b(T,W_1,Z)\nu(T) \right\}= \beta^{-1}\E\{\nu(T^*)\}, 
    \end{align*}
    where the last equality is shown in the proof of Lemma \ref{thm:proximal_trunc_identification}. 
    
    As a special case, take $\nu(t) \equiv 1$ for all $t>0$, we have 
    \begin{align*}
        \E\left\{\frac{\Delta b(X,W_1,Z)}{S_D(X-Q)} \right\} 
        = \beta^{-1}.
    \end{align*}
    Combining the above, the conclusion follows.
    
\end{proof}

\newpage
\section{Model compatibility}\label{app:model_compat}

We confirm the compatibility of the data generating mechanism described in Section \ref{sec:simu}, model \eqref{eq:model_b}, and  Assumptions \ref{ass:proximal_indep} - \ref{ass:completeness_truncbridge}. 
That is, under the data generating mechanism in Section \ref{sec:simu}, 
Assumptions \ref{ass:proximal_indep} - \ref{ass:completeness_truncbridge} are satisfied, and there exists a truncation bridge process satisfying model \eqref{eq:model_b} as well as  conditions \eqref{eq:trunc_bridge_def} and \eqref{eq:trunc_bridge_initial_condi}.
Note that the data generating mechanism described in Section \ref{sec:simu} is a special case for the data generating mechanism in Section \ref{app:plausibility}, so we prove under the latter.

It is straight forward to see that Assumptions \ref{ass:proximal_indep} and \ref{ass:positivity} are satisfied. 
In addition, we have shown in Appendix \ref{app:plausibility} that Assumption \ref{ass:trunc_bridge} is satisfied, and there exists a truncation bridge process satisfying model \eqref{eq:model_b} as well as the conditions \eqref{eq:trunc_bridge_def} and \eqref{eq:trunc_bridge_initial_condi}. 
Assumption \ref{ass:completeness_truncbridge} holds with a similar argument as in \citet{ying2024proximal}.

\newpage
\section{Estimation and asymptotics}

\subsection{Estimator with adjusted IPCW}\label{app:theta_hat_c_adjusted}

% \begin{remark} \label{remark:IPCW}
The estimator $\httc$ in \eqref{eq:theta_hat_c} only include uncensored subjects, i.e., those with $\Delta = 1$; while for some choices of $\nu$, more subjects can be incorporated. 
For instance, when $\nu(t) = \ind(t>t_0)$ or $\nu(t) = \min(t,t_0)$, the value of $\nu(T)$ is known for both uncensored subjects and subjects that are censored after $t_0$. In addition, recall from \eqref{eq:trunc_bridge_initial_condi} that $b(t,W_1,Z) \equiv 1$ for all $t\geq \tauQ$. 
Therefore, we can also include censored subjects with $X>t_0 \vee \tauQ$, by adjusting the IPCW weights using the minimum time at which both $\nu(T)$ and $b(T,W_1,Z)$ are observed \citep{robins1992recovery}. 
Specifically, let $\tilde X = \min(X,t_0 \vee \tauQ)$, and $\tilde\Delta = \ind(\{\Delta = 1\} \cup \{\Delta = 0, t_0 \vee \tauQ<X\})$. 
The identification in Lemma \ref{thm:identification_cen_after} can be adjusted to
    \begin{align}
        \theta =  \left. \E\left\{ \frac{\tilde\Delta b(\tilde X,W_1,Z) \nu(\tilde X) }{S_D(\tilde X - Q)} \right\} \right/
        \E\left\{\frac{\tilde\Delta b(\tilde X,W_1,Z)}{ S_D(\tilde X - Q)} \right\},  \label{eq:identification_cen_2}
    \end{align}
    and the corresponding estimator for $\theta$ is: 
    \begin{align}
    \hat\theta_{\text{adj}} 
    = \left.\left\{\sum_{i=1}^n \frac{\tilde\Delta_i \hat b(\tilde X_i,W_{1i},Z_i)\nu(\tilde X_i)}{\hat S_D(\tilde X_i - Q_i)} \right\} \right/ \left\{\sum_{i=1}^n \frac{\tilde\Delta_i \hat b(\tilde X_i,W_{1i},Z_i)}{\hat S_D(\tilde X_i - Q_i)} \right\}.  \label{eq:theta_tilde_c_adjusted}
\end{align}
    Compared with $\httc$ in \eqref{eq:theta_hat_c}, the consistency and asymptotic normality of $\hat\theta_{\text{adj}}$ require a weaker positivity condition: $S_D(\tilde X - Q) > \etaD$ almost surely for some $\etaD>0$.
    Even when the stronger positivity condition (Assumption \ref{ass:cen_positivity}) holds, $\hat\theta_{\text{adj}}$ tends to be more stable than $\httc$, since the IPCW weights $\tilde\Delta_i/\hat S_D(\tilde X_i-Q_i)$ in $\hat\theta_{\text{adj}}$ have smaller variability 
    %\lily{?}\yuyaox{are closer to one} 
    than the IPCW weights $\Delta_i/\hat S_D(X_i-Q_i)$ in \eqref{eq:theta_hat_c}.
    % In our simulation and data application below, we focus on estimating the marginal survival probabilities, i.e., $\nu(t) = \ind(t>t_0)$. We apply the adjusted IPCW estimator for more stability.
    
    Note that the identification in \eqref{eq:identification_cen_2} and the corresponding estimator for $\theta$ view $\tauQ$ (the supremum for the support of $Q^*$) as known, which is the case when $\tauQ$ can be determined by domain knowledge or study design. When $\tauQ$ is unknown, it can be approximated by $\max_i Q_i$, which is what we used in the simulation and application below.
% \end{remark}

\subsection{Review of the asymptotic results for $U$-statistics}

The following asymptotic results of $U$-statistics from \citet{van2000asymptotic} are used in the proof of Theorem \ref{thm:AN_c} for the asymptotic normality of $\hat\theta$. For completeness, we review the results below. 

Let $X_1,...,X_n$ be a random sample from an unknown distribution and $r$ an integer between 2 and $n$. Given a known function $h(x_1,...,x_r)$, consider the estimation of the parameter $\psi = \E\{h(X_1,...,X_r)\}$. A $U$-statistic with kernel $h$ is defined as
\begin{align*}
U = \frac{1}{{n \choose r}} \sum_{\alpha} h(X_{\alpha_1},...,X_{\alpha_r}),
\end{align*}
where the sum is taken over the set of all unordered subsets $\alpha = (\alpha_1,...,\alpha_r)$ of $r$ different integers chosen from $\{1,...,n\}$. 
The projection of $U-\psi$ onto the set of all statistics of the form $\sum_{i=1}^n g_i(X_i)$ is given by
\begin{align}
\hat U = \sum_{i=1}^n \E(U-\psi|X_i) = \frac{r}{n}\sum_{i=1}^n h_1(X_i),
\end{align}
where 
$h_1(x) = \E\{h(x,X_2,...,X_r)\} - \psi$. 

\begin{proposition}[Theorem 12.3 in \citet{van2000asymptotic}]\label{lem:Ustats}
	If $\E\{h^2(X_1,...,X_r)\}<\infty$, then 
    \[
    n^{1/2} (U- \psi - \hat U)\convp 0.
    \]
    Consequently, the sequence $n^{1/2}(U-\psi)$ is asymptotically normal with mean zero and covariance $r^2\xi_1$, where, with $X_1,...,X_r, X_1',...,X_r'$ denoting i.i.d. random variables, 
    \[
    \xi_1 = \text{cov}(h(X_1,X_2,...,X_r), h(X_1,X_2',...,X_r')).
    \]
\end{proposition}

\newpage
\subsection{Proof of Theorem \ref{thm:AN_c}}\label{app:proof_AN_withC}

\begin{proof}[Proof of Theorem \ref{thm:AN_c}]
   
When $\Delta = 1$, we have $X = T$. So
\begin{align}
    \httc - \theta
    = \left.\left\{\sum_{i=1}^n \frac{\Delta_i \hat b(T_i,W_{1i},Z_i)\{\nu(T_i) - \theta\}}{\hat S_D(T_i - Q_i)} \right\} \right/ \left\{\sum_{i=1}^n \frac{\Delta_i \hat b(T_i,W_{1i},Z_i)}{\hat S_D(T_i- Q_i)} \right\}.  \label{eq:ANc_proof_err_decomp}
\end{align}

(i) We first prove consistency.
We will show that (a) the numerator of \eqref{eq:ANc_proof_err_decomp} converges to zero in probability, and (b) the denominator of \eqref{eq:ANc_proof_err_decomp} converges to $\beta^{-1}$ in probability. Then (a) and (b) imply $\hat\theta - \thetatr \convp 0$.

To show (a), we consider the following decomposition for the numerator of \eqref{eq:ANc_proof_err_decomp}:
\begin{align}
    \frac{1}{n} \sum_{i=1}^n \frac{\Delta_i \hat b(T_i,W_{1i},Z_i)\{\nu(T_i) - \theta\}}{\hat S_D(T_i - Q_i)}
    & = \B_1 + \B_2 + \B_3, \label{eq:thm2_proof_err_decomp}
\end{align}
where 
\begin{align*}
    \B_1 & = \frac{1}{n} \sum_{i=1}^n \frac{\Delta_i}{S_{D}(T_i - Q_i)} \cdot \btr(T_i,W_{1i},Z_i)\{\nu(T_i) - \thetatr\}, \\
    \B_2 & = \frac{1}{n} \sum_{i=1}^n \frac{\Delta_i}{ S_{D}(T_i - Q_i)} \left\{\hat b(T_i,W_{1i},Z_i) - \btr(T_i,W_{1i},Z_i) \right\} \{\nu(T_i) - \thetatr\}, \\
    \B_3 & = \frac{1}{n} \sum_{i=1}^n \left\{\frac{\Delta_i}{\hat S_{D}(T_i - Q_i)} - \frac{\Delta_i}{S_{D}(T_i - Q_i)} \right\}  \hat b(T_i,W_{1i},Z_i)\{\nu(T_i) - \thetatr\};
\end{align*}

In the following, we will show that $\B_1$, $\B_2$, and $\B_3$ all converge to zero in probability.
We first consider $\B_1$. Note that $\B_1$ is an average of $n$ i.i.d. quantities and 
\begin{align*}
    &\quad \E\left[ \frac{\Delta}{S_{D}(T - Q)} \cdot \btr(T,W_{1},Z)\{\nu(T) - \thetatr\} \right] \\
    & =  \E\left(\E\left[\left.  \frac{\ind(T < C) }{S_{D}(T - Q)} \cdot \btr(T,W_{1},Z)\{\nu(T) - \thetatr\} \right| Q,T,W_1,Z \right] \right) \\
    & =  \E\left(  \frac{ \E\left[\left.\ind(T < C) \right| Q,T,W_1,Z \right]}{S_{D}(T - Q)} \cdot \btr(T,W_{1},Z)\{\nu(T) - \thetatr\}\right)
\end{align*}
By Assumption \ref{ass:cen_random}, 
\begin{align*}
    \E\left[\left.\ind(T < C) \right| Q,T,W_1,Z \right]
    & = \E\left[\left.\ind(D > T-Q) \right| Q,T,W_1,Z \right] = S_{D}(T - Q). 
\end{align*}
So 
\begin{align*}
    \E\left[ \frac{\Delta}{S_{D}(T - Q)} \cdot \btr(T,W_{1},Z)\{\nu(T) - \thetatr\}\right] 
    & = \E\left[\btr(T,W_{1},Z)\{\nu(T) - \thetatr\}\right]  \\
    & = \E\left[\btr(T,W_{1},Z)\nu(T) \right] - \thetatr \cdot \E\left[\btr(T,W_{1},Z)\right] \\
    & = 0,
\end{align*}
where the last equation is by  Lemma \ref{thm:proximal_trunc_identification}. % in the proof of Theorem \ref{thm:AN}.
Therefore, by law of large numbers, $\B_1 \convp 0$.

% \red{Following a similar argument as showing $\A_2 \convp 0$ in the proof of Theorem \ref{thm:AN}, we can show that $\B_2\convp 0$.}

For $\B_2$, we have
    \begin{align}
        \E\left(|\B_2|\right)
        &\leq \E\left[ \left|\frac{\Delta}{S_{D}(T - Q)}\right| \cdot \left|\hat b(T,W_{1},Z) - \btr(T,W_{1},Z) \right| \cdot |\nu(T) - \thetatr| \right] \nonumber \\
        & \lesssim \E\left[ \left|\hat b(T,W_{1},Z) - \btr(T,W_{1},Z) \right| \right] \label{eq:AN_proof_1}\\
        & = \left\|\hat b(T,W_{1},Z) - \btr(T,W_{1},Z) \right\|_1 \\
        & = o(1) \label{eq:AN_proof_2},
    \end{align}
    where \eqref{eq:AN_proof_1} holds by Assumption \ref{ass:cen_positivity} and the boundedness of $\nu$, and \eqref{eq:AN_proof_2} holds by Assumption \ref{ass:consistency}.
    Therefore, by Markov's inequality, $\B_2 = o_p(1)$, i.e, $\B_2 \convp 0$. 

We now consider $\B_3$. 
The Kaplan-Meier estimator $\hat S_D$ is known to be uniformly consistent, i.e, $\sup_{t} |\hat S_D(t) - S_{D}(t)| = o_p(1)$.
By the boundedness of $\hat b$ and $\nu$ and  Assumption \ref{ass:cen_positivity}, 
\begin{align*}
    |\B_3| 
    &\lesssim \frac{1}{n} \sum_{i=1}^n \left| \frac{\Delta_i}{\hat S_{D}(T_i - Q_i)} - \frac{\Delta_i}{S_{D}(T_i - Q_i)} \right| \\
    & = \frac{1}{n} \sum_{i=1}^n \frac{\Delta_i \left|\hat S_{D}(T_i -Q_i) - S_{D}(T_i-Q_i) \right|}{\hat S_{D}(T_i-Q_i)S_{D}(T_i-Q_i)} \\
    & \lesssim \sup_{t} |\hat S_D(t) - S_{D}(t)| + o_p(1) \\
    & = o_p(1),
\end{align*}
that is, $\B_3 \convp 0$.  

Combining the above, (a) follows. With a similar argument as above and the result that $\E\{\btr(T,W_{1},Z)\} = \beta^{-1}$ (shown in the proof of Lemma \ref{thm:proximal_trunc_identification}), 
(b) follows.
% it follows that the denominator of \eqref{eq:ANc_proof_err_decomp} converges to $\beta^{-1}$ in probability. 
Therefore, $\httc - \thetatr \convp 0$.

\vspace{1em}
(ii) We now show the asymptotic normality of $\httc$ with the additional Assumption \ref{ass:AL}. 
In the consistency proof, we have shown that the denominator of \eqref{eq:ANc_proof_err_decomp} converges to $\beta^{-1}$ in probability.
By Slutsky's Theorem, it suffices to show that the numerator of \eqref{eq:ANc_proof_err_decomp} is asymptotically normal. 

Again we consider the decomposition \eqref{eq:thm2_proof_err_decomp} for the numerator of \eqref{eq:ANc_proof_err_decomp}. 
Note that $\B_1$ is an average of i.i.d. terms. We will show that $\B_2$ and $\B_3$ are both asymptotically linear. Then the numerator of \eqref{eq:ANc_proof_err_decomp} is asymptotically normal by central limit theorem.

1) We first show that $\B_2$ adopts an asymptotically linear decomposition, that is, there exists a function $h_2(O)$ with $\E\{h_2(O)\} = 0$ and $\E\{h_2(O)^2\} < \infty$ such that $ \B_2 = n^{-1}\sum_{i=1}^n h_2(O_i) + o_p(n^{-1/2})$. 
By Assumption \ref{ass:AL}, 
    \begin{align*}
        \B_2 = \B_{21} + \B_{22},
    \end{align*}
    where
    \begin{align*}
        \B_{21} & = \frac{1}{n^2}\sum_{i=1}^n \sum_{j=1}^n \frac{\Delta_i}{ S_{D}(T_i - Q_i)} \cdot \xi(T_i,W_{1i},Z_i; O_j) \{\nu(T_i) - \thetatr\}, \\
        \B_{22} & = \frac{1}{n}\sum_{i=1}^n  \frac{\Delta_i}{ S_{D}(T_i - Q_i)} \cdot R(T_i,W_{1i},Z_i) \{\nu(T_i) - \thetatr\}.
    \end{align*}
    Since $\nu$ is a bounded function and $|\Delta/S_D(t-Q)| \leq \eta_D^{-1}$ almost surely by Assumption \ref{ass:cen_positivity}, 
    \begin{align*}
        \E\{|\B_{22}|\}
        &\lesssim \E\{|R(T,W_1,Z)|\} 
        = \|R(T,W_1,Z)\|_1 = o(n^{-1/2}). 
    \end{align*}
    By Markov inequality, $\B_{22} = o_p(n^{-1/2})$. 

    Next we show that $\B_{21} = U/2 + o_p(n^{-1/2})$, where $U$ is a $U$-statistic defined in \eqref{eq:thm1_proof_Ustats} below. 
    Then Proposition \ref{lem:Ustats} (the asymptotic results for $U$-statistics) implies that $\B_{21}$ is asymptotically linear. 
    Let
    \begin{align*}
        h(O_i,O_j) 
        & = \frac{\Delta_i}{ S_{D}(T_i - Q_i)}\cdot \xi(T_i,W_{1i},Z_i; O_j) \{\nu(T_i) - \thetatr\} \\
        &\quad + \frac{\Delta_j}{ S_{D}(T_j - Q_j)}\cdot \xi(T_j,W_{1j},Z_j; O_i) \{\nu(T_j) - \thetatr\}.
    \end{align*}
    Then
    \begin{align*}
        \B_{21} 
        & = \frac{1}{n^2} \sum_{i<j} h(O_i,O_j) +  \frac{1}{n^2} \sum_{i=1}^n h(O_i,O_i).
    \end{align*}
    Consider the $U$ statistic
    \begin{align}
        U = \frac{1}{{n\choose 2}} \sum_{i<j} h(O_i,O_j),  \label{eq:thm1_proof_Ustats}
    \end{align}
    where ${n\choose 2} = n(n-1)/2$ is the combinatorial number of $n$ choose 2. 
    We have 
    % \begin{align}
    %     \frac{1}{n^2} = \frac{1}{2}\cdot \frac{1}{{n\choose 2}} - \frac{1}{n^2(n-1)}.
    % \end{align}
    % So
    \begin{align}
        \B_{21} 
        % & = \frac{1}{n^2} \sum_{i<j} h(O_i,O_j) +  \frac{1}{n^2} \sum_{i=1}^n h(O_i,O_i) \nonumber \\
        & = \left\{ \frac{1}{2}\cdot\frac{1}{{n\choose 2}} - \frac{1}{n^2(n-1)} \right\} \sum_{i<j} h(O_i,O_j) + \frac{1}{n^2} \sum_{i=1}^n h(O_i,O_i) \nonumber  \\
        & = \frac{U}{2} - \frac{1}{n^2(n-1)} \sum_{i<j} h(O_i,O_j) + \frac{1}{n^2} \sum_{i=1}^n h(O_i,O_i) 
    \end{align}
    In the following we show that the last two terms in the above equation are both $o_p(n^{-1/2})$.
    We have 
    \begin{align*}
        & \E\left\{\left|\frac{1}{n^2(n-1)} \sum_{i<j} h(O_i,O_j)\right|\right\}
        \leq \frac{1}{n^2(n-1)}\sum_{i<j} \E\left\{\left|h(O_i,O_j)\right|\right\} = O(n^{-1}). \\
        & \E\left\{\left|\frac{1}{n^2} \sum_{i=1}^n h(O_i,O_i)\right|\right\}
        \leq \frac{1}{n^2}\sum_{i=1}^n \E\left\{\left|h(O_i,O_i)\right|\right\} = O(n^{-1}).
    \end{align*}
    By Markov inequality,  
    \begin{align*}
        & \frac{1}{n^2(n-1)} \sum_{i<j} h(O_i,O_j) 
        = O_p(n^{-1}) = o_p(n^{-1/2}), \\
        & \frac{1}{n^2} \sum_{i=1}^n h(O_i,O_i)
        = O_p(n^{-1}) = o_p(n^{-1/2}).
    \end{align*}
    Therefore, 
    \begin{align}
        \B_{21} 
        & = \frac{U}{2} + o_p(n^{-1/2}).  \label{eq:AN_proof_4}
    \end{align}
    
    Before applying Proposition \ref{lem:Ustats} to show $U$ is asymptotically linear,
    we first compute $\E\{h(O_1,O_2)\}$, which corresponds to the $\psi$ involved in Proposition \ref{lem:Ustats}. Recall from Assumption \ref{ass:AL} that $\xi$ is the influence function of $\hat b$.
    By Assumption \ref{ass:cen_random}, 
    \begin{align*}
        \E\left\{\left. \frac{\Delta_1}{S_{D}(T_1 - Q_1)} \right| T_1, Q_1, W_{11}, Z_1, O_2\right\}
         = \E\left\{\left. \frac{\ind(D_1>T_1-Q_1)}{S_{D}(T_1 - Q_1)} \right| T_1, Q_1 \right\} = 1.
    \end{align*}
    Therefore, by the tower property of expectations, 
    \begin{align}
        \E\{h(O_1,O_2)\} 
        & = \E\left[ \frac{\Delta_1}{ S_{D}(T_1 - Q_1)}\cdot \xi(T_1,W_{11},Z_1; O_2) \{\nu(T_1) - \thetatr\}\right] \\
        &\quad + \E\left[ \frac{\Delta_2}{ S_{D}(T_2 - Q_2)}\cdot \xi(T_2,W_{12},Z_2; O_1) \{\nu(T_2) - \thetatr\}\right] \nonumber \\
        & = \E\left( \E\left[\left. \frac{\Delta_1}{ S_{D}(T_1 - Q_1)}\right| T_1,Q_1,W_{11},Z_1,O_2\right] \xi(T_1,W_{11},Z_1; O_2) \{\nu(T_1) - \thetatr\}  \right)\\
        &\quad + \E\left( \E\left[\left. \frac{\Delta_2}{ S_{D}(T_2 - Q_2)} \right| T_2,Q_2,W_{12},Z_2,O_1\right] \xi(T_2,W_{12},Z_2; O_1) \{\nu(T_2) - \thetatr\}\right) \nonumber \\
        & = \E\left[\xi(T_1,W_{11},Z_1; O_2) \{\nu(T_1) - \thetatr\}\right]
        + \E\left[\xi(T_2,W_{12},Z_2; O_1) \{\nu(T_2) - \thetatr\}\right] \nonumber \\
        & = \E\left[ \E\left\{\xi(T_1,W_{11},Z_1; O_2)\mid O_1\right\} \{\nu(T_1) - \thetatr\}\right] \nonumber \\
        &\quad 
        + \E\left[\E\left\{\xi(T_2,W_{12},Z_2; O_1)\mid O_2\right\} \{\nu(T_2) - \thetatr\}\right] \nonumber \\
        & = 0, \label{eq:AN_proof_5}
    \end{align}
    where the last equation holds by the definition of influence function and the independence between $O_1$ and $O_2$.
    Therefore, Proposition \ref{lem:Ustats} implies 
    \begin{align*}
        U = \frac{2}{n} \sum_{i=1}^n h_2(O_i) + o_p(n^{-1/2}),
    \end{align*}
    where $h_2(o) = \E\{h(o,O)\}$ and $o = (q,t,w_1,w_2,z)$.
    This together with \eqref{eq:AN_proof_4} implies 
    \begin{align*}
        \B_{21} & = \frac{1}{n} \sum_{i=1}^n h_2(O_i) + o_p(n^{-1/2}).
    \end{align*}
    Combining the above, 
    \begin{align}
        \B_{2} & = \B_{21} + \B_{22} 
        = \frac{1}{n} \sum_{i=1}^n h_2(O_i) + o_p(n^{-1/2}), \label{eq:AN_proof_B2}
    \end{align}
    
    Lastly, we verify that $\E\{h_2(O)\} = 0$ and $\E\{h_2(O)^2\}<\infty$.
    Recall that $h_2(o) = \E\{h(o,O)\}$. 
    Since $O_1$ and $O_2$ are independent and each has the same distribution as $O$, by definition, 
    \begin{align*}
        h_2(O_1) = \E\{h(O_1,O_2) \mid O_1\}. 
    \end{align*}
    Therefore, 
    \begin{align*}
        \E\{h_2(O)\} 
        = \E\{h_2(O_1)\}
        = \E[\E\{h(O_1,O_2)\mid O_1\}]
        = \E\{h(O_1,O_2)\} 
        =0,
    \end{align*}
    where the last equation holds by \eqref{eq:AN_proof_5}. 
    Next we show that $\E\{h_2(O)^2\}<\infty$. 
    Note that 
    \begin{align*}
        h_2(O_1) 
        & = \E\{h(O_1,O_2)\mid O_1\} \\
        & = \E\left[\left. \xi(T_1,W_{11},Z_1; O_2) \right| O_1 \right] \cdot \frac{\Delta_1}{S_{D}(T_1 - Q_1)} \cdot \{\nu(T_1) - \thetatr\} \\
        &\quad + \E\left[\left.\frac{\Delta_2}{S_{D}(T_2 - Q_2)}\cdot \xi(T_2,W_{21},Z_2; O_1) \{\nu(T_2) - \thetatr\} \right| O_1 \right].
    \end{align*}
    By Assumption \ref{ass:cen_positivity} and the boundedness of $\nu$,  
    \begin{align}
        \E\{h_1(O)^2\}  
        & = \E\{h_1(O_1)^2\} \\
        & \lesssim \E\left[\left.\E\left\{\xi(T_1,W_{11},Z_1; O_2) \right| O_1\right\}^2 \right]  \nonumber\\
        &\quad + \E\left( \E\left[\left.\frac{\Delta_2}{S_{D}(T_2 - Q_2)}\cdot \xi(T_2,W_{21},Z_2; O_1) \{\nu(T_2) - \thetatr\} \right| O_1 \right]^2 \right)\\
        &\leq \E\left[\left.\E\left\{\xi(T_1,W_{11},Z_1; O_2)^2 \right| O_1\right\} \right]  \nonumber\\
        &\quad + \E\left( \E\left[\left.\left|\frac{\Delta_2}{S_{D}(T_2 - Q_2)}\right|^2\cdot \xi(T_2,W_{21},Z_2; O_1)^2 \{\nu(T_2) - \thetatr\}^2 \right| O_1 \right]\right) \label{eq:AN_proof_6} \\
        &\lesssim  \E\left\{\xi(T_1,W_{11},Z_1; O_2)^2 \right\}, \label{eq:AN_proof_7} \\
        & <\infty. \label{eq:AN_proof_10}
    \end{align}
    where \eqref{eq:AN_proof_6} holds by Jensen's inequality; \eqref{eq:AN_proof_7} holds by the boundedness of $\nu$, $\Delta/S_D(T-Q)\leq \eta_D^{-1}$ almost surely by Assumption \ref{ass:cen_positivity}, as well as the symmetric role of $O_1$ and $O_2$; and \eqref{eq:AN_proof_10} holds by Assumption \ref{ass:AL}.

\vspace{0.5em}
2) We now show that $\B_3$ also adopts an asymptotic linear decomposition.
We have
\begin{align*}
    \B_3  = \B_{31} + \B_{32},
\end{align*}
where
\begin{align*}
    \B_{31} & =  \frac{1}{n} \sum_{i=1}^n \left\{\frac{\Delta_i}{\hat S_{D}(T_i-Q_i)} - \frac{\Delta_i}{S_{D}(T_i-Q_i)} \right\} \btr(T_i,W_{1i},Z_i)\{\nu(T_i) - \thetatr\}, \\
    \B_{32} & =  \frac{1}{n} \sum_{i=1}^n \left\{\frac{\Delta_i}{\hat S_{D}(T_i -Q_i)} - \frac{\Delta_i}{S_{D}(T_i -Q_i)} \right\} \left\{\hat b(T_i,W_{1i},Z_i) - \btr(T_i,W_{1i},Z_i) \right\}\{\nu(T_i) - \thetatr\}. 
\end{align*}
Since the Kaplan-Meier estimator $\hat S_D$ is asymptotically linear with remainder term being  $o_p(n^{-1/2})$ uniformly in $t$ \citep{cai1998asymptotic}.
% This together with Assumption \ref{ass:cen_positivity} imply that $\Delta/\hat S_{D}(T-Q)$ is asymptotically linear. 
Using a similar proof as that for showing $\B_{2}$ is asymptotically linear, we can show that $\B_{31}$ is asymptotically linear,
i.e., there exists $h_3(O)$ with $\E\{h_3(O)\} = 0$ and $\E\{h_3(O)^2\} < \infty$ such that 
\begin{align*}
    \B_{31} & = \frac{1}{n} \sum_{i=1}^n h_3(O_i) + o_p(n^{-1/2}).
\end{align*}
Lastly, we show $\B_{32} = o_p(n^{-1/2})$. Since $\hat S_{D}$ and $\hat b$ are both asymptotically linear, 
with a similar proof as that in \citet[proof of Theorem S2 in the Supplementary Material for showing $B_2 = o_p(n^{-1/2})$]{wang2024doubly},
we can show that $\B_{32} = o_p(n^{-1/2})$. 
Combining the above, 
\begin{align}
    \B_{3} & = \frac{1}{n} \sum_{i=1}^n h_3(O_i) + o_p(n^{-1/2}). \label{eq:AN_proof_B3}
\end{align}

Finally, combining \eqref{eq:thm2_proof_err_decomp} \eqref{eq:AN_proof_B2} \eqref{eq:AN_proof_B3}, we have 
\begin{align}
    \frac{1}{n} \sum_{i=1}^n \frac{\Delta_i \hat b(T_i,W_{1i},Z_i)\{\nu(T_i) - \thetatr\} }{ \hat S_D(T_i-Q_i)}
    & = \frac{1}{n} \sum_{i=1}^n \phi(O_i) + o_p(n^{-1/2}),  \label{eq:AN_proof_11} 
\end{align}
where 
\begin{align*}
    \phi(O) = \frac{\Delta\btr(T,W_{1},Z)\{\nu(T) - \thetatr\}}{S_{D}(T- Q)} + h_2(O) + h_3(O).
\end{align*}
In the following, we verify that $\E\{\phi(O)\} = 0$ and $\E\{\phi(O)^2\} < \infty$, which together with central limit theorem and Slusky's theorem conclude the proof. 
By Lemma \ref{thm:identification_cen_after}, 
\begin{align*}
    E\left[\frac{\Delta\btr(T,W_{1},Z)\{\nu(T) - \thetatr\}}{S_{D}(T- Q)}\right] 
    = E\left\{\frac{\Delta\btr(T,W_{1},Z)\nu(T)}{S_{D}(T- Q)}\right\}  - \thetatr \cdot E\left\{\frac{\Delta\btr(T,W_{1},Z)}{S_{D}(T- Q)}\right\} = 0. 
\end{align*}
In addition, since $\nu$ is a bounded function, and $\Delta/S_D(T-Q)\leq \eta_D^{-1}$ almost surely by Assumption \ref{ass:cen_positivity}, 
\begin{align*}
    E\left[\frac{\Delta\btr(T,W_{1},Z)\{\nu(T) - \thetatr\}}{S_{D}(T- Q)}^2 \right] \lesssim \|b(T,W_1,Z)\|_{\sup} < \infty. 
\end{align*}
These together with $\E\{h_2(O)\} = \E\{h_3(O)\} = 0$, $\E\{h_2(O)^2\}<\infty$ and $\E\{h_3(O)^2\}<\infty$ imply that
$\E\{\phi(O)\} = 0$ and $\E\{\phi(O)^2\} < \infty$. 
Therefore, by central limit theorem, \eqref{eq:AN_proof_11} implies 
\begin{align*}
    \frac{1}{\sqrt{n}} \sum_{i=1}^n \frac{\Delta_i \hat b(T_i,W_{1i},Z_i)\{\nu(T_i) - \thetatr\} }{ \hat S_D(T_i-Q_i)} \convd N(0, \E\{\phi(O)^2\}). 
\end{align*}
The conclusion follows with $\sigmac^2 = \beta^2 \E\{\phi(O)^2\}$ by Slutsky's Theorem.
    
\end{proof}

\clearpage
\section{Additional details for HAAS data analysis}\label{app:HAAS}

\subsection{Data pre-processing for grip strength} \label{app:HAAS_preprocessing}

For grip strength, each subject was allowed 3 attempts for each hand, and we take the maximum of the measurements \citep{charles2006occupational}. As shown in Figure \ref{fig:HAAS_GripX1_AgeX1}, grip strength shows a clear declining trend with age, and the trend looks linear. 
Therefore, we standardized it by taking the residuals after fitting a linear model of grip strength on age.

\begin{figure}[h]
\centering
 \includegraphics[width=0.6\textwidth]{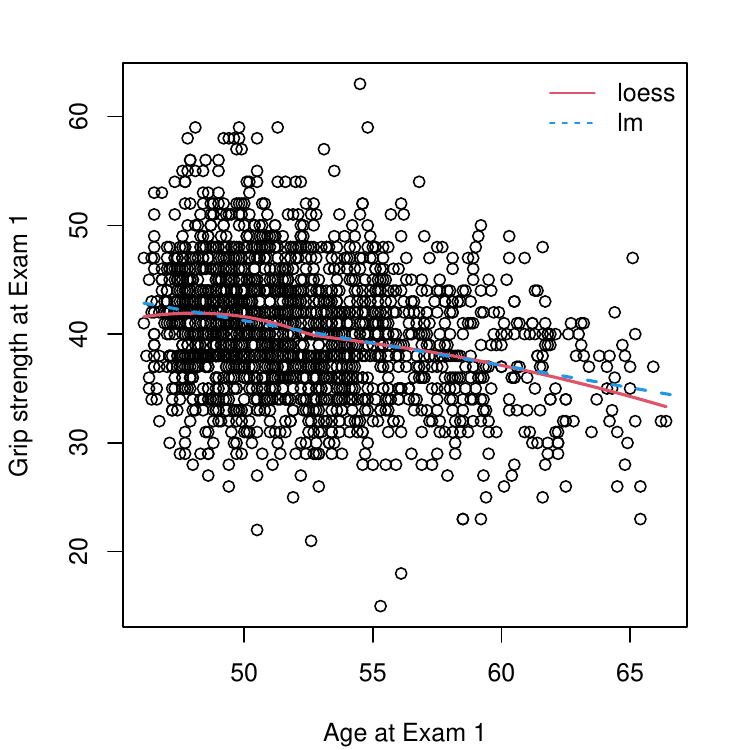}
\caption{Scatter plot, loess fit, and linear fit for grip strength measured at HHP baseline versus age at HHP baseline.}
\label{fig:HAAS_GripX1_AgeX1}
\end{figure}

\clearpage
\subsection{Covariate distributions}\label{app:HAAS_cov_distr}

\begin{table}[ht]
\centering
\caption{Covariate distribution in HAAS data.}
\label{tab:table1}
\begin{tabular}{lc}
  \toprule
  & Overall ($n=$1930) \\ 
  \midrule
  Education - count (\%)\\
    \quad $> 12$ years  & 418 (21.7) \\

  {\it APOE} genotype  - count (\%)\\
    \quad {\it APOE E4} positive & 369 (19.1) \\ 
  
  Alcohol consumption - count (\%)\\
    \quad Heavy drinker at Exam 1 & 272 (14.1) \\
    \quad Heavy drinker at Exam 3 & 366 (19.0) \\

  Cigarettes consumption - count (\%)\\
    \quad Smoker & 1371 (71.0) \\

  Systolic blood pressure - mean (SD) & 149.47 (21.97) \\ 
  
  Heart rate - mean (SD) & 31.56 (4.69) \\ 
  
  Grip strength (standardized) - mean (SD) & 0.00 (5.66) \\ 
  \bottomrule
\end{tabular}
\end{table}

\end{document}